\documentclass[a4paper,11pt]{article}

\usepackage{microtype} 
\usepackage{authblk}
\RequirePackage{footmisc}

\usepackage{xspace}
\usepackage{algorithm2e}
\usepackage[utf8]{inputenc}
\usepackage{hyperref}
\usepackage{amsthm,amsfonts,amssymb,amsmath,array}
\usepackage{graphicx}
\usepackage{fullpage}

\usepackage{thmtools,thm-restate}
\usepackage{cleveref}

\newtheorem{theorem}{Theorem}
\newtheorem{observation}[theorem]{Observation}
\newtheorem{lemma}[theorem]{Lemma}

\newtheorem{proposition}[theorem]{Proposition}

\newtheorem{definition}[theorem]{Definition}

\let\realbfseries=\bfseries
\def\bfseries{\realbfseries\boldmath}

\newcommand{\marrow}{\marginpar[\hfill$\longrightarrow$]{$\longleftarrow$}}

\newif\ifComments
\Commentstrue

\usepackage{xcolor}
\ifComments
  \newcommand{\says}[2]{\textcolor{red}{\textsc{#1 says:} \marrow\textsf{#2}}}
\else
  \newcommand{\says}[2]{\relax}
\fi

\newcommand{\MIT}{leapfrog}
\newcommand{\aux}{musketeer}
\newcommand{\auxes}{musketeers}
\newcommand{\Aux}{Musketeer}
\newcommand{\p}{$\Gamma$\ }
\newcommand{\pa}{I\ }
\newcommand{\pat}{Z\ }

\newcommand{\facet}{facet}
\newcommand{\Facet}{Facet}
\newcommand{\afacet}{a facet}

\title{
Universal Reconfiguration of \Facet-Connected Modular Robots by Pivots:
The $O(1)$ Musketeers}

\author[1]{Hugo A. Akitaya\thanks{Email: hugo.alves\_akitaya@tufts.edu. Supported by NSF CCF-1422311 and CCF-1423615.}}
\author[2]{Esther M. Arkin\thanks{Email: esther.arkin@stonybrook.edu. Partially funded by by NSF (CCF-1526406).}}
\author[3]{Mirela Damian\thanks{Email: mirela.damian@villanova.edu}}
\author[4]{Erik D. Demaine\thanks{Email: edemaine@mit.edu. Supported in part by NSF ODISSEI grant EFRI-1240383 and NSF Expedition grant CCF-1138967.}}
\author[5]{Vida Dujmovi\'c\thanks{Email: vida.dujmovic@uottawa.ca}}
\author[6]{Robin Flatland\thanks{Email: flatland@siena.edu}}
\author[1]{Matias Korman\thanks{Email: Matias.Korman@tufts.edu}}
\author[7]{Bel\'en Palop\thanks{Email: belen.palop@uva.es. Partially supported by MTM2015-63791-R (MINECO/FEDER).}}
\author[8]{Irene Parada\thanks{Email: iparada@ist.tugraz.at. Supported by the Austrian Science Fund (FWF): W1230.}}
\author[9]{Andr\'e van Renssen\thanks{Email: andre.vanrenssen@sydney.edu.au. Supported by JST ERATO Grant Number JPMJER1201, Japan.}}
\author[10]{Vera Sacrist\'an\thanks{Email: vera.sacristan@upc.edu. Partially supported by MTM2015-63791-R (MINECO/FEDER) and Gen. Cat. DGR 2017SGR1640.}}

\affil[1]{Tufts University, USA}
\affil[2]{State University of New York at Stony Brook, USA}
\affil[3]{Villanova University, USA}
\affil[4]{Massachusetts Institute of Technology, USA}
\affil[5]{University of Ottawa, Canada}
\affil[6]{Siena College, USA}
\affil[7]{Universidad de Valladolid, Spain}
\affil[8]{Graz University of Technology, Austria}
\affil[9]{The University of Sydney, Australia}
\affil[10]{Universitat Polit\`ecnica de Catalunya, Spain}

\begin{document}

\maketitle

\vspace{-2em}
\begin{abstract}
	We present the first universal reconfiguration algorithm
	for transforming a modular robot between any two \facet-connected
	square-grid configurations using pivot moves.
	More precisely, we show that five extra ``helper'' modules (``musketeers'')
	suffice to reconfigure the remaining $n$ modules
	between any two given configurations.
	Our algorithm uses $O(n^2)$ pivot moves, which is worst-case optimal.
	Previous reconfiguration algorithms either require less restrictive
	``sliding'' moves, do not preserve \facet-connectivity,
	or for the setting we consider, could only handle a small
	subset of configurations defined by a local forbidden pattern.
	Configurations with the forbidden pattern do have disconnected
	reconfiguration graphs (discrete configuration spaces), and indeed
	we show that they can have an exponential number of connected components.
	But forbidding the local pattern throughout the configuration is far
	from necessary, as we show that just a constant number of added modules
	(placed to be freely reconfigurable)
	suffice for universal reconfigurability.
	We also classify three different models of natural pivot moves that
	preserve \facet-connectivity, and show separations between these models.
\end{abstract}

\section{Introduction}\label{sec:intro}

\emph{Shape shifting} is a powerful idea in science fiction:
T-1000 robots (from \textit{Terminator 2: Judgement Day}),
Changelings (from \textit{Star Trek: Deep Space 9}),
Symbiotes (from \textit{Venom}),
Mystique (from \textit{X-Men}), and
Metamorphagi (from \textit{Harry Potter})
all have the ability to transform their shape nearly arbitrarily.
How can we make shape shifting into science?

\emph{Modular robots} \cite{survey2017,self-reconfigurable,survey2007}
are perhaps the best answer to this question.
The idea is to build a single ``robot'' out of many small units called
\emph{modules}, each of which can attach and detach from each other,
move relative to each other, communicate with each other, and compute.
Modular robots offer extreme adaptability to changing environment
or user needs, in particular by reconfiguring the modules into
exponentially many effective shapes of the overall robot.
Modularity also offers a practical future for manufacturing
(identical modules can be mass-produced, making them relatively cheap), 
makes robots easy to repair by just replacing the broken modules,
and makes it possible to re-use components from one robot/task to another.

For computational geometry, modular robots offer exciting challenges:
what shapes can a modular robot self-reconfigure into, and what are
good algorithms for reconfiguration?
According to \cite{self-reconfigurable},
the main difficulties in self-reconfiguration
are the physical motion constraints of the modules themselves,
connectivity requirements for the robot to hold together,
collisions between moving and/or static modules, and ``deadlocks''
where no module can move or some module gets ``trapped''
within the configuration.

The wide diversity of mecatronic solutions to modular robots can be
characterized from a geometric viewpoint by three key properties:
(1) the lattice, (2) connectivity requirement, and
(3) allowed moves.

\paragraph{Lattice.}
Most modular robots follow a space-filling lattice structure
(e.g., squares or hexagons in 2D, or cubes in 3D),
to simplify both reconfiguration and
the characterization of possible shapes.
Pure lattice modular robots
\cite{EMCube,Metamorphic,3DFractum,Fractum,Crystal,M-blocks}
have one robot per lattice element and always remain on the lattice,
while hybrid modular robots \cite{M-tran,Atron,SuperBot,Molecube}
also allow units move out of the lattice.
We focus here on the well-studied square lattice, though we suspect our
results can be generalized to cube lattices.

\paragraph{Connectivity requirement.}

A modular robot generally needs to be connected at all times while
reconfiguring, so that the modules do not fall apart.
The most common and practical constraint is that the modules are always
\emph{\facet-connected}, meaning a connected \emph{\facet-adjacency graph}
where vertices represent modules and edges represent adjacencies by
shared facets (edges in 2D).
The exception is that the moving module is excluded from this graph
during each move, meaning that other modules must be \facet-connected
while the moving module may briefly disconnect during the move.
A weaker connectivity constraint,
considered in some theoretical research \cite[Ch.~4]{nadia},
is that the robot is connected via shared vertices.
In such case, reconfiguration is always possible.
We focus here on the more challenging facet-connectivity constraint.

\paragraph{Allowed moves.}

One of the most popular models is \emph{sliding squares/cubes}
\cite{pushing-cubes,pushing-squares,MeltGrow},
illustrated in Figure~\ref{fig:slide-pivot} (left).
In this case, modules live in a square or cube lattice,
move by sliding relative to each other, and require \facet-connectivity.
For this model, universal reconfiguration is possible between any two
\facet-connected configurations, in any dimension
\cite{pushing-cubes,pushing-squares}. 
\begin{figure}[hbt]
	\centering
	\includegraphics[page=1,width=.6\textwidth]{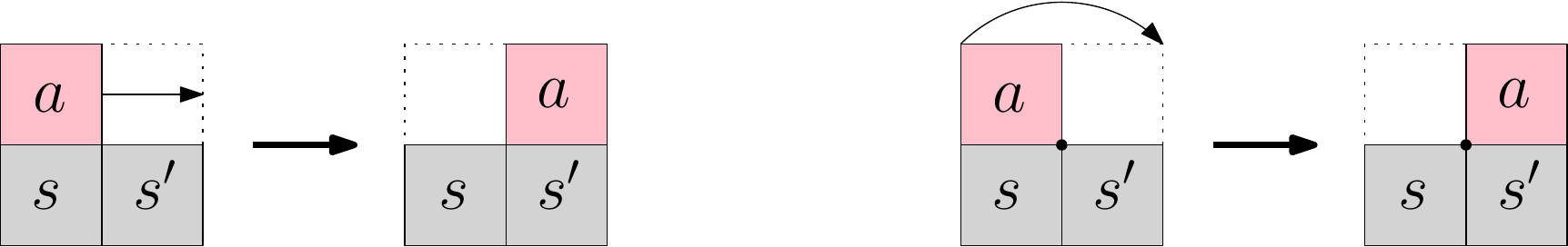}
	\caption{Two ways a module $a$ starting above module $s$ can move to the adjacent lattice position, above module $s'$.
		Left: sliding. Right: pivoting.
		Pivoting requires more free space to execute.}
	\label{fig:slide-pivot}
\end{figure}

We focus here on a more challenging model,
\emph{pivoting squares/cubes} \cite{nadia,M-blocks,ICubes},
illustrated in Figure~\ref{fig:slide-pivot} (right).
In this case, modules live in a square or cube lattice, move by rotating
relative to each other, and require facet-connectivity.
The key difference is that a module needs two additional squares/cubes
of empty space in order to pivot, whereas a slide just needs the
destination square/cube to be empty.
Unfortunately, some configurations are \emph{rigid} in this model,
meaning that no module can move without disconnecting the robot. 

Rigid configurations appear also in the sliding square model when the sliding capability is restricted to turning corners~\cite{PSPACE-sliding-corners}. 
However, in this model the existence of free space around the modules does not guarantee reconfigurability, 
while in the pivoting squares model it does, 
as we will show later. 

As a consequence, all known reconfiguration algorithms for pivoting
squares/cubes are somehow partial.
One algorithm follows some heuristics without a termination guarantee
\cite{heuristics-square}
(see also \cite{heuristics} for heuristics for hexagons).
A recent algorithm guarantees reconfiguration by forbidding one
or more local patterns in both the start and goal configurations
\cite{M-blocks}, essentially preventing narrow holes in the shape.
(A similar result was obtained for hexagons \cite{density}.)
These assumptions severely restrict the possible shapes
that can be reconfigured, to a $o(1)$ fraction. 
The absence of such local patterns though 
is far from being necessary for reconfigurability. 
In 3D, some further strong conditions are added,
such as that every hole must be orthogonally convex~\cite{M-blocks}.

\paragraph{Our results.}

Our main result is that \emph{universal} reconfiguration is possible if we allow the
addition of a constant number of (five) extra ``helper'' modules,
which we call \emph{musketeer modules}.%
\footnote{\textit{The Three Musketeers} is a story about four musketeers.
	This paper is a story about five musketeers.}
A similar idea was recently applied to a slightly different model of programmable matter in~\cite{PSPACE-sliding-corners}, where helpers are called \emph{seeds}. 
The key is that these helper modules are not considered part of the
initial or target shape, and thus we are free to place them where we like
(in particular, along the external boundary of the robot).
Surprisingly, this small amount of additional freedom is enough to achieve
universal reconfiguration.
In fact, we prove in Section~\ref{sec:algorithm} that
five musketeer modules are both sufficient and sometimes
necessary to solve any reconfiguration under our strategy.
Our algorithm is based on the old idea of following the
right-hand rule to escape a maze~\cite{right-hand}.
The number of pivoting moves it makes is $O(n^2)$, which is optimal in the
worst case by an earth-moving lower bound:
each robot may need to move a distance of $\Theta(n)$.

This result can be seen as proving connectivity of the
\emph{reconfiguration graph} ${\cal G}_{n,k}$, where vertices represent
\facet-connected configurations of $n$ modules and edges represent
valid pivot moves, with the addition of $k \geq 5$ musketeer modules.
With $k = 0$ musketeers, ${\cal G}_{n,k}$ is known to be disconnected.
Surprisingly, there have been no (successful) attempts to understand the
structure of this reconfiguration graph.
In Section~\ref{sec:reconfigraph}, we analyze the structure of this
reconfiguration graph.
Specifically, we prove that ${\cal G}_{n,0}$ can have an
exponential number of connected components of exponential size,
and in some models, can have an exponential number of singleton
connected components (rigid configurations); while in other models,
the reconfiguration graph cannot have any singleton connected components.

One other main contribution of this paper is to precisely define a
variety of natural models for pivot moves.
Pivoting is naturally defined as the rotation of one module
about one of its vertices that is shared with a (static) module. 
But there are some subtleties in this definition depending on exactly
which modules must be \facet-connected at what times.
(Obviously, the moving module is not \facet-connected to the
others during the move.)
In Section~\ref{sec:setting}, we define three nested models, each
at least as powerful as the previous,
and in Section~\ref{sec:separation}, we prove strict separations
between these models. 
Our analysis of connected components in the reconfiguration space
(in Section~\ref{sec:reconfigraph})
also considers the effects of these different models.
We conclude with open problems in Section~\ref{sec:conclusions}.

\section{Models and definitions}\label{sec:setting}

\subsection{Pivot moves}

In a square grid, the fact that two squares may share a vertex without actually sharing an edge opens a wider range of possibilities for the pivoting move. Refer to Figure~\ref{fig:square-pivot}.
\begin{figure}[bth]
	\centering
	\includegraphics[page=2,width=.6\textwidth]{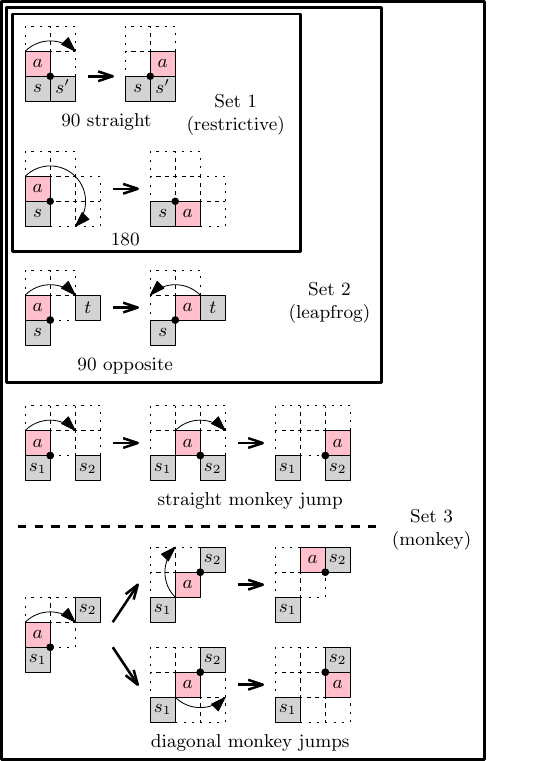}
	\caption{The possible sets of moves for a pivoting module $a$ about a module $s$, in a square grid.}
	\label{fig:square-pivot}
\end{figure}
The most restrictive set of moves (Set 1 in Figure~\ref{fig:square-pivot}) requires module $a$ to be \facet-adjacent to module $s$ and to rotate about one of the two vertices of the edge they share. Such move can be a $90^\circ$ or a $180^\circ$ rotation, depending on whether or not $s$ has a neighboring module $s'$ adjacent to it through the other edge of $s$ incident to the rotation center, and of course, requires the goal grid position to be empty and some intermediate positions to be (at least partially) clear. These cells are depicted in white in Figure~\ref{fig:square-pivot}.

The authors of \cite{M-blocks} propose an expanded set of moves (Set 2 in Figure~\ref{fig:square-pivot}) that allows module $a$ to rotate $90^\circ$ about module $s$ even when $s'$ is not present, as long as module $a$ is again \facet-adjacent to another module $t$ at the end of the move. Since their reconfiguration algorithm relies on reversible moves, this implies allowing also the reverse move: module $a$ can rotate $90^\circ$ about a vertex of another module $s$ incident to $a$, without requiring $s$ to be \facet-adjacent to $a$, as long as $a$ is \facet-adjacent to some module before performing the move and after performing the move. We call this enlarged set the \emph{leapfrog} set of moves.

If the previous move is allowed (i.e., if it is feasible for a given modular robot prototype), it seems natural to allow concatenating more than one of such moves, i.e., to allow concatenating consecutive rotations about vertices incident to the pivoting module. It is easy to prove that such concatenation cannot involve more than two pivots
before the moving module becomes \facet-adjacent to another module. Indeed, if a module $a$ is \facet-adjacent to a module $s_1$, after at most two such moves it necessarily becomes adjacent to a module $s_2$ (Set 3 in Figure~\ref{fig:square-pivot}). We call this complete set the \emph{monkey} set of moves.

It is worth noticing that diagonal monkey jumps may be unnecessary for most purposes, since a diagonal monkey jump can be simulated by ignoring the presence of module $s_2$ and keep rotating about $s_1$ and its neighboring modules. Naturally, this would imply performing a higher number of pivoting moves. We further discuss this issue in Section~\ref{sec:outer-shell}.   

\subsection{Reconfiguration problem}
Consider a configuration $C$ of $n$ robot modules in a given grid. 
The \emph{\facet-adjacency graph} of $C$ has a node for each module, and an edge
between a pair of nodes if the corresponding modules are \facet-adjacent. Throughout this paper we will often refer to the \facet-adjacency graph simply as the
\emph{adjacency graph}. We will say that a configuration $C$ is \emph{\facet-connected} if the adjacency graph of $C$ is connected.

Applying a pivot move from one of the three sets of moves described in the previous section to \afacet-connected configuration $C$, means applying one of the moves to a module in $C$, in such a way that the configuration (without the pivoting module) stays \facet-connected before, after, and during the move, and the pivoting module does not collide with any other module.
Note that this implies that even after deleting the moving module the configuration remains \facet-connected.
Reconfiguring $C$ consists of applying a concatenation of such moves. 

The (universal) reconfiguration problem asks whether it is possible to reconfigure any \facet-connected configuration of $n$ modules in a given grid into any other configuration with the same number of modules.

For any positive integer $n$, the \emph{reconfiguration graph} ${\cal G}_n$ has a node for each \facet-connected configuration with $n$ modules, and an edge between two nodes if the corresponding configurations can be reconfigured into each other through a single pivoting move.  
We call \emph{rigid} any configuration in which no module can move, i.e., any configuration that is an isolated node of ${\cal G}_n$, forming a connected component that is a singleton. We call \emph{locked} any configuration that cannot be reconfigured into a straight strip of modules, i.e., any configuration belonging to a connected component of ${\cal G}_n$ that does not contain a strip.

\section{Reconfiguration graph}\label{sec:reconfigraph}

Figure~\ref{fig:square-rigid} (left) shows an example of a configuration that is rigid under the largest possible set of pivoting moves (Set 3 in Figure~\ref{fig:square-pivot}).
\begin{figure}
	\centering
	\includegraphics[page=3,width=.25\textwidth]{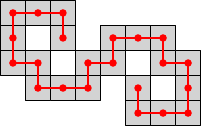}
	\hspace{5em}
	\qquad
	\includegraphics[page=4,width=.33\textwidth]{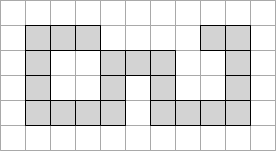}
	\caption{Left: a rigid configuration of edge-connected pivoting squares.
		Right: A configuration that can be reconfigured into a strip, in spite of containing instances of the three forbidden patterns.}
	\label{fig:square-rigid}
\end{figure}
In~\cite{M-blocks} it is proved that reconfiguration
for Set 2 of pivoting moves (leapfrog moves) is possible between two \facet-connected configurations of the same number of squares, provided that they are both \emph{admissible} shapes.
Admissibility is defined in terms of forbidden patterns: \afacet-connected configuration of squares is admissible if it does not contain any of the three patterns---\p\hspace{-0.222 em}, \pa\hspace{-0.222 em}, and \pat\hspace{-0.222 em}---depicted in Figure~\ref{fig:square-pattern}.
However, this local separation condition is certainly not necessary, as proved by the example in Figure~\ref{fig:square-rigid} (right). 

\begin{figure}[tb]
	\centering
	\includegraphics[page=5,width=\textwidth]{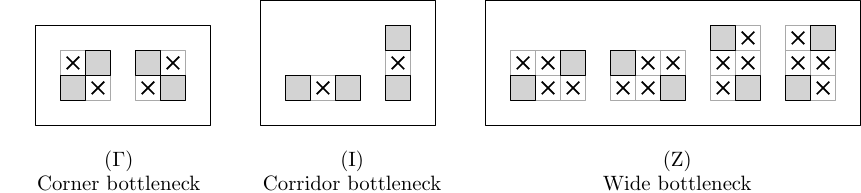}
	\caption{The three forbidden patterns for \facet-connected pivoting squares; solid squares
		represent modules, and $\times$-ed squares represent empty spaces.}
	\label{fig:square-pattern}
\end{figure}

These results raise several natural questions for \facet-connected pivoting squares: Are the three sets of moves equivalent? In particular, is reconfigurability between admissible shapes also guaranteed when using the most restrictive set of pivoting moves?
This latter question has been answered positively by the results from~\cite{M-blocks}. Although not explicitly stated, the reconfiguration algorithm from~\cite{M-blocks} uses only restrictive moves. 

Several other interesting questions are open.
Can the admissible condition be relaxed when using the largest set of pivoting moves? Do there exist rigid configurations that contain only one type of pattern? If so, are they rigid with respect to all three sets of pivoting moves? What can we say about the reconfiguration graph ${\cal G}_n$ for the different sets of pivoting moves? We try to answer these questions in the remaining of this section. 

\subsection{Separation between the different sets of moves}
\label{sec:separation}
We start by showing that the three sets of moves for pivoting squares are not equivalent, as they produce three different reconfiguration graphs.

\begin{proposition}
	\label{prop:prop1}
	The monkey set of moves for pivoting squares (Set 3) is stronger than the leapfrog set (Set 2), and
	the leapfrog set is stronger than the restrictive set (Set 1).
	That is, the resulting reconfiguration graph ${\cal G}_n$
	has strictly fewer connected components for Set 3 than for Set 2,
	and fewer connected components for Set 2 than for Set 1.
\end{proposition}

\begin{proof}
	To show the first inequality consider the two configurations from Figure~\ref{fig:moves-different-1-2}, which include a single module that  can pivot without disconnecting the configuration (shaded pink).
	This module can pivot along some piece of the boundary that is different depending on the pivoting moves allowed.
	If the leapfrog pivoting moves (Set 2) are used, 
	the two configurations belong to the same connected component of  ${\cal G}_n$,
	but not if only the restrictive pivoting moves (Set 1) are allowed.
	\begin{figure}[tb]
		\centering
		\includegraphics[page=6,width=.8\textwidth]{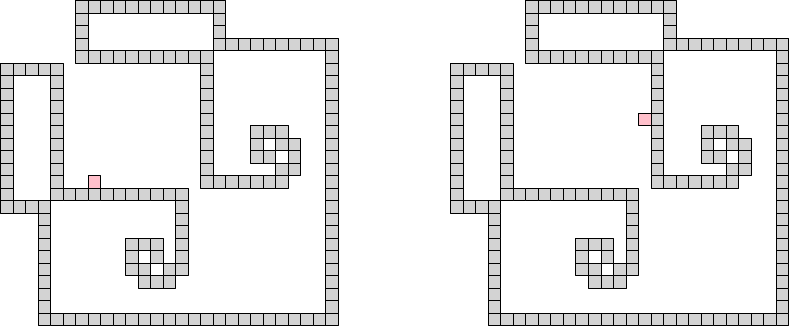}
		\caption{Two shapes that can be reconfigured into each other using the leapfrog set of moves, but
			cannot be reconfigured into each other using the restrictive set of moves.}
		\label{fig:moves-different-1-2}
	\end{figure}
	For the second inequality consider the configuration in Figure~\ref{fig:moves-different-2-3}.
	It is rigid for the leapfrog set of moves (Set 2), but it can easily be reconfigured into a strip using the monkey moves (Set 3).
	\begin{figure}[tb]
		\centering
		\includegraphics[page=7,width=.6\textwidth]{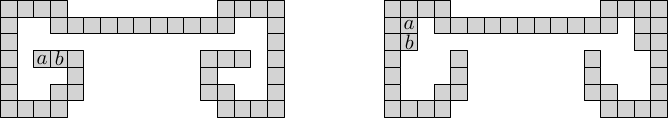}
		\caption{A shape that can be reconfigured into a strip after moving modules $a$ and $b$ using the monkey set of moves, but it cannot be reconfigured using the leapfrog set of moves.}
		\label{fig:moves-different-2-3}
	\end{figure}
\end{proof}

Let us now discuss the differences between the three forbidden patterns (depicted in Figure~\ref{fig:square-pattern}). 
From a purely geometric viewpoint, pattern \p produces a (corner) bottleneck along the boundary of a configuration that is narrower than the one produced by pattern \pa (corridor bottleneck). This one is in turn narrower than the one produced by pattern \pat (wide bottleneck).
The next propositions show how the presence or the absence of each of such patterns influences reconfiguration under each of the 3 sets of pivoting moves.

\subsection{Pattern \p\hspace{-0.222 em}: corner bottleneck}
We start by showing that pattern \p alone suffices to make a configuration rigid, regardless of the set of pivoting moves used (restrictive, leapfrog, or monkey).

\begin{proposition}
	Let ${\cal G}_n$ be the reconfiguration graph of \facet-connected pivoting squares.
	If only pattern \p is allowed, while patterns \pa and \pat are forbidden,
	the number of connected components of ${\cal G}_n$ that are singletons
	and the number of connected components of ${\cal G}_n$ of exponential size are both exponential, regardless of the set of pivoting moves used.
\end{proposition}

\begin{proof}
	Figure~\ref{fig:pattern-1-locked} (top) shows a rigid configuration.
	Notice that the dark gray path 
	can be configured in $\Omega(2^n)$ ways, since each pair of consecutive dark-gray modules can be connected at least in two different ways (East-West and North-South, for example).
	
	We can modify this construction to obtain the locked configuration from Figure~\ref{fig:pattern-1-locked}, where each of the pink modules can pivot inside a 
	hole. 
	First,  no matter where the pink modules sit, none of the gray modules can move.
	Therefore, if the number of such inner holes of the configuration is $k$, the size of the corresponding connected component of ${\cal G}_n$ is $\Omega(8^k)$.
	Moreover, each of the light pieces creating an inner hole has at most 72 modules.
	Second, if the path of dark modules has length $p$, the number of different connected components that are obtained is $\Omega(2^p)$, since  again
	each pair of consecutive dark modules can be connected at least in two different ways. 
	The constructions 
	Making $k := \lfloor\varepsilon n / 73\rfloor$ and $p := \lfloor(1-\varepsilon)n\rfloor$
	for any $0 < \varepsilon < 1$, we have that asymptotically there are $\Omega(2^{n})$ components of $2^{\Omega(n)}$ size.
	\begin{figure}[tb]
		\centering
		\includegraphics[page=8,width=.7\textwidth]{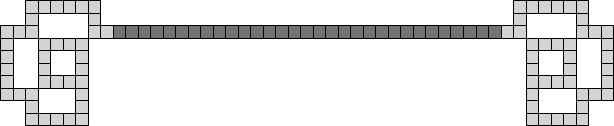}
		\bigskip
		
		\includegraphics[page=9,width=.8\textwidth]{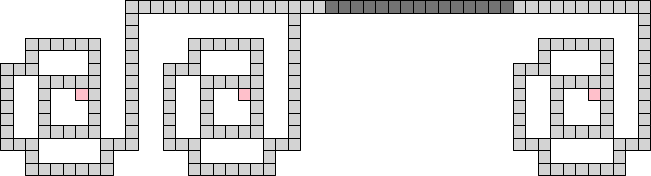}
		\caption{Configurations showing only the forbidden pattern \p\hspace{-0.222 em}.
			Top: Rigid configuration for any set of pivoting moves.
			Bottom: Locked configuration for any set of pivoting moves.}
		\label{fig:pattern-1-locked}
	\end{figure}
\end{proof}

\subsection{Pattern \pa\hspace{-0.222 em}: corridor bottleneck}
The forbidden pattern \pa is weaker than pattern \p in the sense that it suffices to make a configuration rigid for the sets of moves 1 and 2 (restrictive and \MIT). However, if the entire Set 3 of moves is allowed, pattern \pa alone cannot make a configuration rigid, as shown by Proposition~\ref{prop:iSing} below.

\begin{proposition}
	Let ${\cal G}_n$ be the reconfiguration graph of \facet-connected pivoting squares under sets 1 and 2 of pivoting moves (restrictive and leapfrog).
	If only pattern \pa is allowed, and patterns \p and \pat are forbidden,
	the number of connected components of ${\cal G}_n$ that are singletons
	and the number of connected components of ${\cal G}_n$ of exponential size are both exponential.
\end{proposition}

\begin{proof}
	Figure~\ref{fig:pattern-2-locked} shows a configuration containing only pattern \pa that is rigid 
	under sets 1 and 2 of pivoting rules.
	The path of dark modules can have different shapes,
	as each pair of consecutive dark modules can be \facet-adjacent in at least two ways.
	Thus, a rigid construction without the part containing the pink modules
	can be configured in $\Omega(2^n)$ ways.
	If we included the gadgets containing the pink modules, 
	we obtain a locked configuration, where each of the pink modules can pivot to three different positions.
	Therefore, if the number of such gadgets is $k$, the size of the corresponding connected component of ${\cal G}_n$ is $\Omega(3^k)$. 
	Moreover, if the path of dark modules has length $p$, the number of different connected components that are obtained is $\Omega(2^p)$. 
	The right and the left ends of the configuration shown in Figure~\ref{fig:pattern-2-locked} consist of 22 modules each.  
	Additionally, there are 2 modules that lie between the dark modules and the gadgets. 
	Making $k := \lfloor\varepsilon (n-46) / 7\rfloor$ and $p := \lfloor(1-\varepsilon)(n-46)\rfloor$
	for any $0 < \varepsilon < 1$ we have that asymptotically there are $\Omega(2^{n})$ components of $2^{\Omega(n)}$ size.
\end{proof}
\begin{figure}[tbh]
	\centering
	\includegraphics[page=10,width=.7\textwidth]{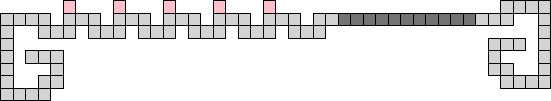}
	\caption{A configuration showing only pattern \pa (corridor bottleneck). It is locked for both the restrictive and the leapfrog pivoting moves.}
	\label{fig:pattern-2-locked}
\end{figure}

In contrast, if the entire set of monkey-pivoting moves is allowed, then no configuration can be rigid if it only contains instances of pattern \pa (and no instance of patterns \p and \pat\hspace{-0.222 em}).

\begin{proposition}
	\label{prop:iSing}
	Let ${\cal G}_n$ be the reconfiguration graph of \facet-connected pivoting squares under the entire Set 3 of monkey-pivoting moves. If only pattern \pa is allowed, and patterns \p and \pat are forbidden, then ${\cal G}_n$ contains no singleton components.
\end{proposition}

In order to prove this result, we make use of two definitions and a lemma.

\begin{definition}
	A corner of \afacet-connected configuration of squares is a module that is adjacent to at least two empty grid positions through two consecutive edges: North and East, South and East, North and West, or South and West.
\end{definition}

\begin{definition}
	Let $G$ be the \facet-adjacency graph of a given \facet-connected configuration of squares. The cactus graph $T(G)$ of $G$ is defined as follows. For each simple cycle $C$ in $G$, consider the set $Region(C)$ of grid positions that lie in $C$ or are enclosed by $C$. A simple cycle $C$ is said to be maximal if $Region(C)$ is maximal with respect to inclusion. We define $T(G)$ as the connected subgraph of $G$ that contains all the leaves of $G$, all the maximal cycles of $G$, and all the connections among them.
\end{definition}

\begin{figure}[tbh]
	\centering
	\includegraphics[page=11,width=.8\textwidth]{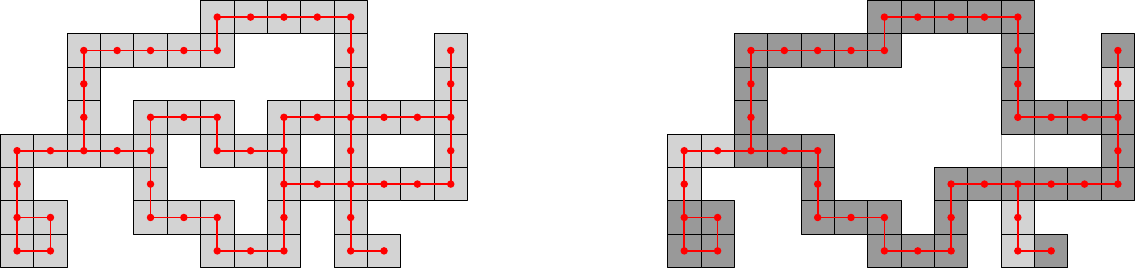}
	\caption{The adjacency graph $G$ of \afacet-connected configuration (left) and the corresponding cactus graph $T(G)$ (right); leaves and maximal cycles are dark-shaded.}
	\label{fig:cactus-graph}
\end{figure}

Figure~\ref{fig:cactus-graph} illustrates this definition.
Although without a specific name or designation, this same graph has been previously used in~\cite{pushing-squares} to help prove the connectedness of the reconfiguration graph of modular robots that follow the sliding cube model.

\begin{lemma}\label{lem:corner}
	Let $C$ be any configuration of \facet-connected squares. Let $G$ be the \facet-adjacency graph of $C$. There always exists a corner in $C$ that is not a cut vertex of $G$.
\end{lemma}

\begin{proof}
	Let $T(G)$ be the cactus graph of $G$.
	If $T(G)$ has a node $m$ of degree one, then $m$ corresponds to a corner in $C$ that is not a cut vertex of $G$, so the lemma holds. Otherwise, we view $T(G)$ as a tree of cycles, and arbitrarily pick a leaf cycle in $T(G)$ (note that at least one such leaf cycle exists). The leaf cycle must have a corner different from the (unique) node that connects it to the rest of $T(G)$. Such a node cannot be a cut vertex. 
\end{proof}

We can now proceed to prove Proposition~\ref{prop:iSing}.

\begin{proof}[Proof of Proposition~\ref{prop:iSing}]
	Let $C$ be a configuration of \facet-connected pivoting squares that does not contain patterns \p and \pat\hspace{-0.222 em}. By Lemma~\ref{lem:corner} there exists a module $c$ in $C$ such that \emph{i)} $c$ is a corner and \emph{ii)} removal of $c$ does not disconnect $C$. We will prove that $c$ can pivot. Assume, without loss of generality, that $c$ is a North-East corner, and that it is connected to the rest of $C$ through its South; refer to Figure~\ref{fig:pattern-2-non-rigid}.
	\begin{figure}[h]
		\centering
		\includegraphics[page=12,width=.95\textwidth]{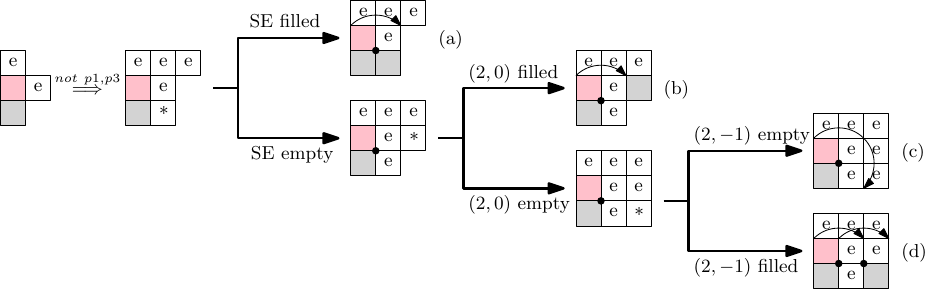}
		\caption{If only pattern \pa is allowed some corner can pivot using monkey-pivoting moves.}
		\label{fig:pattern-2-non-rigid}
	\end{figure}
	Since $c$ is a North-East corner, the grid positions to its East and North are empty. Since patterns \p and \pat do not appear in $C$, the grid positions $(1,1)$ and $(2,1)$ relative to $c$ are also empty. There are two possibilities for the grid position South-East from $c$. If it is occupied by a module, then $c$ can pivot $90^\circ$; see Figure~\ref{fig:pattern-2-non-rigid} (a).
	If it is empty, then again there are two possibilities for the grid position $(2,0)$ relative to $c$. If it is occupied, then $c$ can pivot $90^\circ$; see Figure~\ref{fig:pattern-2-non-rigid} (b).
	If it is empty, consider the grid position $(2,-1)$ with respect to $c$. If this is empty too, then $c$ can pivot $180^\circ$; see Figure~\ref{fig:pattern-2-non-rigid} (c).
	If it is occupied, then $c$ can make a
	straight monkey jump; see Figure~\ref{fig:pattern-2-non-rigid} (d).
\end{proof}

\subsection{Pattern \pat\hspace{-0.222 em}: wide bottleneck}
The forbidden pattern \pat is weaker than the forbidden patterns \p and \pa in the sense that no configuration can be rigid if it contains only instances of pattern \pat
\hspace{-0.222 em}. 

\begin{proposition}
	Let ${\cal G}_n$ be the reconfiguration graph of \facet-adjacent pivoting squares. If only pattern \pat is allowed, and patterns \p and \pa are forbidden, then ${\cal G}_n$ contains no singleton components, regardless of the set of pivoting moves allowed.
\end{proposition}

\begin{proof}
	Let $C$ be a configuration of \facet-connected pivoting squares that does not contain patterns \pa and \pat\hspace{-0.222 em}. By Lemma~\ref{lem:corner} there exists a module $c$ in $C$ such that \emph{i)} $c$ is a corner and \emph{ii)} removal of $c$ does not disconnect $C$. We will prove that $c$ can pivot. Assume, without loss of generality, that $c$ is a North-East corner, and that it is connected to the rest of $C$ through its South. Since $c$ is a North-East corner, the grid positions to its East and North are empty. Since pattern \p does not appear in $C$, the grid position to the North-East of $c$ is also empty. Finally, there are two possibilities for the grid position South-East from $c$. If it is occupied by a module, then $c$ can pivot $90^\circ$, as illustrated in Figure~\ref{fig:pattern-3-non-rigid} (top). If it is empty, then the positions $(2,0)$ and $(2,-1)$ relative to $c$ need to be empty as well, since pattern \pa is forbidden; see Figure~\ref{fig:pattern-3-non-rigid} (bottom). Therefore, $c$ can pivot $180^\circ$.
	\begin{figure}[tbh]
		\centering
		\includegraphics[page=13,width=.7\textwidth]{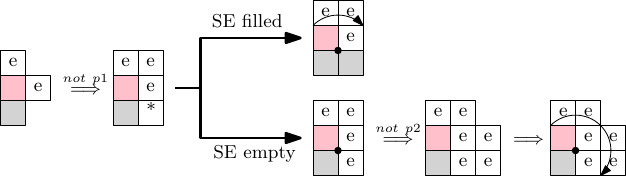}
		\caption{If only pattern \pat is allowed, any corner can pivot regardless of the set of pivoting moves used.}
		\label{fig:pattern-3-non-rigid}
	\end{figure}
\end{proof}

However, there can be locked configurations containing only instances of pattern \pat\hspace{-0.222 em}. 
Figure~\ref{fig:pattern-3-locked} (top) shows two configurations that are locked for Set 1 of pivoting moves (restrictive) and do not have any instance of patterns \p or \pa\hspace{-0.222 em}, but only instances of pattern \pat (wide bottleneck). 

\begin{figure}[tbh]
	\centering
	\includegraphics[page=14,width=0.6\textwidth]{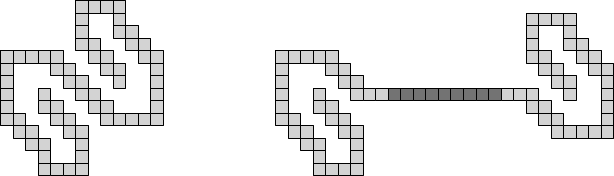}
	\bigskip
	
	\includegraphics[page=15,width=\textwidth]{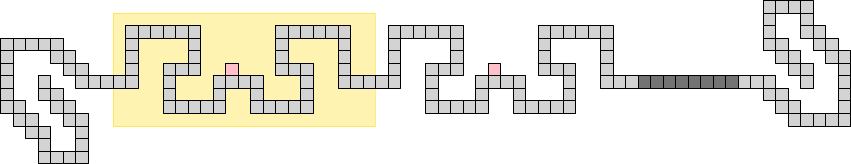}
	
	\caption{Three configurations with only instances of pattern \pat\hspace{-0.222 em}. They cannot be reconfigured into a strip by means of the pivoting moves in Set 1 (restrictive).}
	\label{fig:pattern-3-locked}
\end{figure}

\begin{proposition}
	\label{prop:zExp}
	Let ${\cal G}_n$ be the reconfiguration graph of \facet-connected pivoting squares under pivoting set of moves 1.
	If only pattern \pat is allowed, and patterns \p and \pa are forbidden,
	the number of connected components of ${\cal G}_n$ of exponential size is
	exponential.
\end{proposition}

\begin{proof}
	Figure~\ref{fig:pattern-3-locked} (bottom) shows a configuration containing only pattern \pat that is locked
	under Set 1 of pivoting rules.
	Each of the pink modules can pivot to three different positions inside the gadget (highlighted in Figure~\ref{fig:pattern-3-locked}) containing it, which consists of 53 modules.
	Therefore, if the number of such gadgets is $k$, the size of the corresponding connected component of ${\cal G}_n$ is $\Omega(4^k)$.
	Moreover, if the path of dark modules has length $p$, the number of different connected components that are obtained is $\Omega(2^p)$,
	as each pair of consecutive dark modules can be \facet-adjacent in at least two ways. 
	The right and the left ends of the configuration shown in Figure~\ref{fig:pattern-2-locked} consist of 36 modules each. 
	Making $k := \lfloor\varepsilon (n-72) / 53\rfloor$ and $p := \lfloor(1-\varepsilon)(n-72)\rfloor$
	for any $0 < \varepsilon < 1$ we have that asymptotically there are $\Omega(2^{n})$ components of $2^{\Omega(n)}$ size.
\end{proof}

\section{Universal reconfiguration algorithm with $O(1)$ musketeers}
\label{sec:algorithm}

In this section, we aim for the important practical goal of universal
reconfiguration, that is, connectivity of the reconfiguration graph.
We have seen that the local separation condition (while sufficient) is too strong: 
Robot configurations can contain many instances of the forbidden patterns
and still be reconfigurable.
On the other hand, we proved that as soon as the local separation
condition is relaxed, the reconfiguration graph breaks into at least an
exponential number of connected components of exponential size.

In what follows, we propose and analyze a new approach for reconfiguring arbitrary \facet-connected configurations (which may contain an arbitrary number of instances of the forbidden patterns).
Our strategy is based on the addition of $O(1)$ \emph{musketeer modules}, i.e.,  modules that can freely move around the boundary
of our robot configuration and will be used as helpers in certain situations.
These modules are not necessarily part of the specified initial or target configuration.

\subsection{Preliminaries: outer shell\label{sec:outer-shell}}

Let $C$ be an arbitrary \facet-connected configuration of pivoting squares.
We start by introducing a few definitions.

Let $G$ be the \facet-adjacency graph of $C$, and $\overline{G}$ the \facet-adjacency graph of the lattice cells that are not occupied by a module of $C$.
Each bounded connected component of $\overline{G}$ is a \emph{hole} of the robot configuration $C$.
The only unbounded connected component of $\overline{G}$ is the \emph{exterior} of $C$.
The \emph{boundary} of $C$ is the set of lattice cells that are empty and are \facet-adjacent to (at least) one module of $C$.
If the configuration has holes, we define its \emph{external boundary}
as the subset of the boundary contained in the unbounded connected component of $\overline{G}$.

\begin{lemma}
	\label{lem:righthand}
	Let $C$ be an arbitrary and static \facet-connected configuration of pivoting squares.
	Let $m$ be an active module attached to $C$, North of the topmost rightmost module of $C$.
	Using the monkey set of moves (Set 3), $m$ can pivot along the external boundary of $C$ following the right-hand rule and return to its initial position.
	If only the leapfrog set of moves (Set 2) is allowed, this is not always possible.
\end{lemma}

\begin{proof}
	We start by proving that such a sequence of moves from Set 3 is always possible. In order to do that, we prove that an invariant holds before and after each possible move, which allows $m$ to pivot clockwise following the right-hand rule.
	
	Let the cell positions be denoted by their relative coordinates with respect to $m$. The invariant is that position $(0,1)$ is empty and that if at least one of $a_0=(-1,0)$ or $a_1=(-1,1)$ is occupied, then both $b_0=(1,0)$ and $b_1=(1,1)$ must be empty; see Figure~\ref{fig:RightHandInvariant}. Naturally, the invariant will be rotated as $m$ moves along the boundary of $C$, so that it is always facing outward. The invariant is trivially satisfied initially by our choice of starting position.
	\begin{figure}[tbh]
		\centering
		\includegraphics[page=16,width=.1\textwidth]{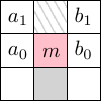}
		\caption{Illustration of the invariant. The active module $m$ (in pink) is \facet-adjacent to a static module of $C$ (in gray). The opposite cell (striped) is empty. If $a_1$ or $a_2$ are occupied, then both $b_1$ and $b_2$ are empty.}
		\label{fig:RightHandInvariant}
	\end{figure}
	
	We want to show that, when the invariant is satisfied, $m$ can pivot according to the right-hand rule and that the position it moves to satisfies the invariant. First consider the case where $b_0$ is occupied and thus $a_0$ and $a_1$ are empty. If $b_1$ or $c_2=(0,2)$ are occupied, $m$ can pivot to $(0,1)$ and the invariant holds in a rotated version; see Figure~\ref{fig:RightHandCase1}.
	\begin{figure}[tbh]
		\centering
		\includegraphics[page=17,width=.75\textwidth]{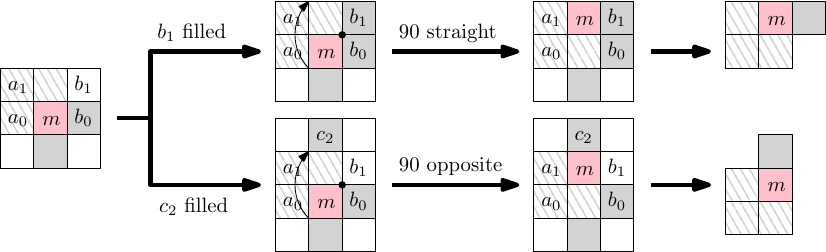}
		\caption{When $b_0$ is occupied, $a_0$ and $a_1$ are empty. The move, when $b_1$ or $c_2$ are occupied.
			Striped cells are empty, filled cells are occupied.}
		\label{fig:RightHandCase1}
	\end{figure}
	
	On the other hand, when $b_1$ and $c_2$ are empty, consider $b_2=(1,2)$; refer to Figure~\ref{fig:RightHandCase2}. If $b_2$ is empty, $m$ can move to $b_1$ and maintain the invariant. If $b_2$ is occupied, consider $a_2=(-1,2)$. If this cell is empty, $m$ can move there and after pivoting the invariant is satisfied (in a rotated version). Finally, if $b_2$ and $a_2$ are both occupied, $m$ can move to $a_1$ and a rotated version of the invariant holds.
	
	\begin{figure}[tbh]
		\centering
		\includegraphics[page=18,width=.95\textwidth]{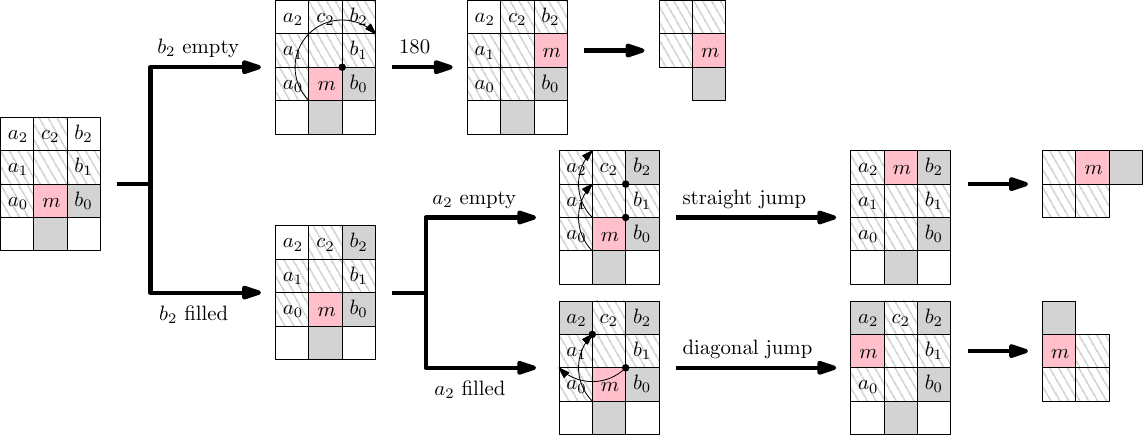}
		\caption{When $b_0$ is occupied, $a_0$ and $a_1$ are empty. The move, when $b_1$ or $c_2$ are both empty.
			Striped cells are empty, filled cells are occupied.}
		\label{fig:RightHandCase2}
	\end{figure}
	
	Now, consider the case where $b_0$ is empty, illustrated in Figure~\ref{fig:RightHandCase3}. If $b_1$ is occupied, the invariant guarantees that $a_1$ and $a_2$ are empty, and $m$ can move to $b_1$. After the move, a rotated version of the invariant holds.
	\begin{figure}[tbh]
		\centering
		\includegraphics[page=19,width=.7\textwidth]{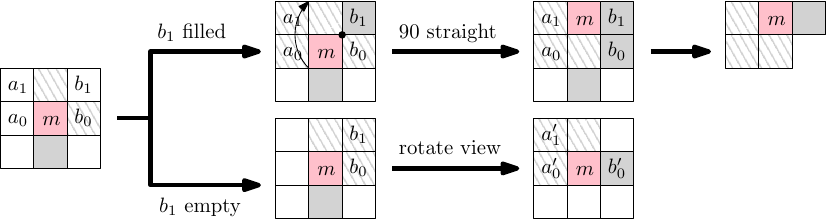}
		\caption{The move, when $b_0$ is empty.
			Striped cells are empty, filled cells are occupied.}
		\label{fig:RightHandCase3}
	\end{figure}
	Finally, notice that the case where both $b_0$ and $b_1$ are empty is equivalent to the case where $b_0$ is occupied. This is made visibly clear after rotating the entire configuration $90^\circ$.
	
	Since the static configuration $C$ has a finite number of modules and is \facet-connected, this concludes the proof that, using the monkey set of moves (Set 3) $m$ can pivot along the external boundary of $C$ following the right-hand rule and return to its initial position. In addition, since the starting position of $m$ belongs to the external boundary of $C$, and no move can get $m$ into a hole, the positions occupied by $m$ along its traversal are all located in the external boundary.
	
	If the robot is only allowed the operations of the leapfrog model, then it is not always possible for the active module $m$ to pivot along the boundary and return to its initial position. Indeed, any configuration containing one of the two situations depicted in Figure~\ref{fig:RightHandCase2}, where the moving module performs a monkey jump, cannot be traversed in the leapfrog model. In the leapfrog model, the active module $m$ would get stuck and could not advance following the right-hand rule.
\end{proof}

It is worth noticing that the proof of Lemma~\ref{lem:righthand} does not require the use of diagonal monkey jumps, but only of straight monkey jumps. 
This is relevant form a practical viewpoint, 
since it allows our results to be applied to a larger class of modular robots. 
For example, the hardware systems modeled in~\cite{heuristics-square,M-blocks} can perform straight monkey jumps, but not diagonal ones. 

We can now define the \emph{outer shell} of \afacet-connected configuration $C$ of pivoting squares to be the subset of the external boundary of $C$ formed by the lattice cells eventually occupied by any active robot module $m$ initially positioned North of the topmost-rightmost module in $C$, in its right-hand rule traversal of the boundary of $C$, described in Lemma~\ref{lem:righthand}. Figure~\ref{fig:outside} illustrates this concept.
\begin{figure}[tbh]
	\centering
	\includegraphics[page=20,width=.8\textwidth]{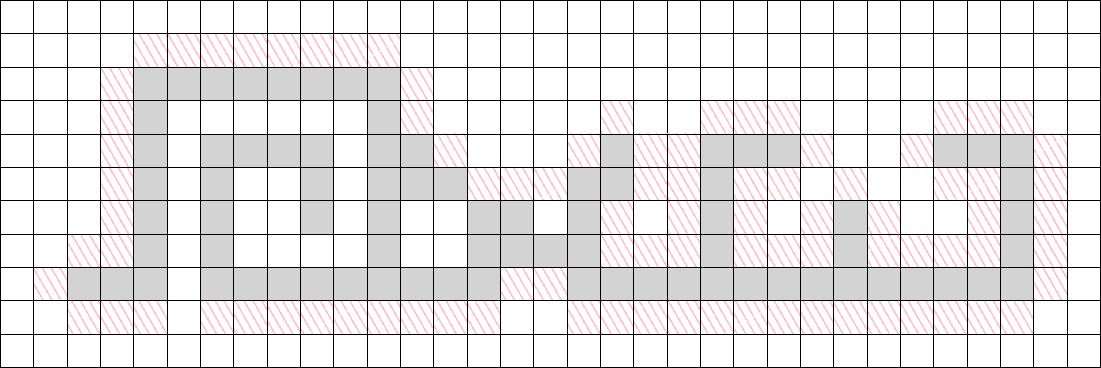}
	\caption{A robot configuration (in gray) and its associated outer shell (striped in pink).}
	\label{fig:outside}
\end{figure}

\subsection{Algorithm Overview}
Our reconfiguration algorithm transforms any initial \facet-connected configuration $C$ of pivoting squares into any goal configuration with the same number of modules.

In order to simplify the algorithm's description, we use an intermediate canonical configuration, say a strip, and describe the transformation from the initial shape to the strip. The reconfiguration from the strip to the final shape is obtained by reversing the steps of the algorithm. The strip can be built from any lexicographically best positioned module of the configuration. For example, we will grow a horizontal strip to the left of the bottommost of the leftmost modules of the configuration.

The strategy behind the algorithm is simple.
It consists of sequentially choosing a module from the configuration that is not a cut vertex of its \facet-adjacency graph, and make it pivot, following the right-hand rule, along the outer shell, until it reaches the tip of the strip and stops.
The problem of this strategy, as we saw in Section~\ref{sec:reconfigraph}, is that the reconfiguration graph is not connected, even under the extended set of monkey moves.
In order to overcome this problem, the algorithm uses \emph{\aux\ modules}. 
Any module from the canonical strip can serve as a \aux\ module.
We will prove that five \aux\ modules are sufficient and sometimes necessary to solve any reconfiguration based on our strategy.
Because the canonical strip is initially empty, it may be necessary to add \aux\ modules to the strip, if fewer than needed are available (this may happen at most once).

\subsection{Algorithm Details}

The description of the algorithm and the proof of its correctness make use of a {\em potential} function. If $m$ is a module located in the lattice position with coordinates $(x,y)$, the potential function at $m$ is defined as $\Phi(m)=(x+y,x)$.
The potential being a two-dimensional function, we sort its values lexicographically.
The maximum potential $\Phi_{max}$ (minimum potential $\Phi_{min}$) of a configuration is the lexicographically largest (smallest) potential of all its modules.
Note that, whenever we use the term \emph{configuration}, we refer to the \facet-connected component that includes all modules other than the ones in the canonical strip.

Given any configuration, we define North-East and South-West as being the modules with highest and lowest potential, respectively. 
Notice that in any configuration $C$, both North-East and South-West are \facet-adjacent to the outer shell of $C$.

\subsubsection{\Aux\ modules}

\begin{definition}
	We say that a module $m$ is {\em outer-free} in a configuration $C$ if it is \facet-adjacent to the outer shell and it can pivot clockwise, without disconnecting the robot.
\end{definition}

The first step of the algorithm is to look for an outer-free module, and pivot it to the tip of the strip. This step is repeated until no further outer-free modules exist in the initial configuration. If at that point the configuration is a strip, the algorithm ends.

Otherwise, all the modules in the strip may be used as \aux\ modules, one at at time, starting from the tip of the strip, pivoting them to the positions where they are needed, as described in next Section~\ref{sec:bridging}. 
Since our algorithm may require five such modules, it may be necessary to add extra modules to the strip (or anywhere in the configuration where they are outer-free) 
in order to complete the necessary set of \aux\ modules. This can be done at this stage or on the fly, as needed. This second option may be preferable in some cases, as not all configurations require as many as five \aux\ modules.

\subsubsection{Bridging procedure}\label{sec:bridging}
In this section we describe an operation necessary in some situations when there are no outer-free modules in the configuration.
Let $m$ be the North-East module, i.e., the maximum potential module of a given configuration $C$.
Trivially, there can be no modules of $C$ located North, North-East, or East of $m$, i.e., in positions $(0,1)$, $(1,1)$, or $(1,0)$ relative to $m$.
Therefore, the degree of $m$ in the \facet-adjacency graph can only be 1 or 2.
Since $m$ is not outer-free, it must be a cut vertex and have degree 2.
Let $b_1$ and $g_1$ respectively be its counterclockwise and clockwise \facet-adjacent modules; see Figure~\ref{fig:3-square}. 
We color the two connected components of $C$ connected by $m$ blue and green, so that $b_1$ is blue and $g_1$ is green. 
One important procedure of our algorithm, which we call \emph{bridging}, is the act of using \aux\ modules to connect the green and blue components so that $m$ becomes outer-free. 

\begin{figure}[tbh]
	\centering
	\includegraphics[page=21,width=\textwidth]{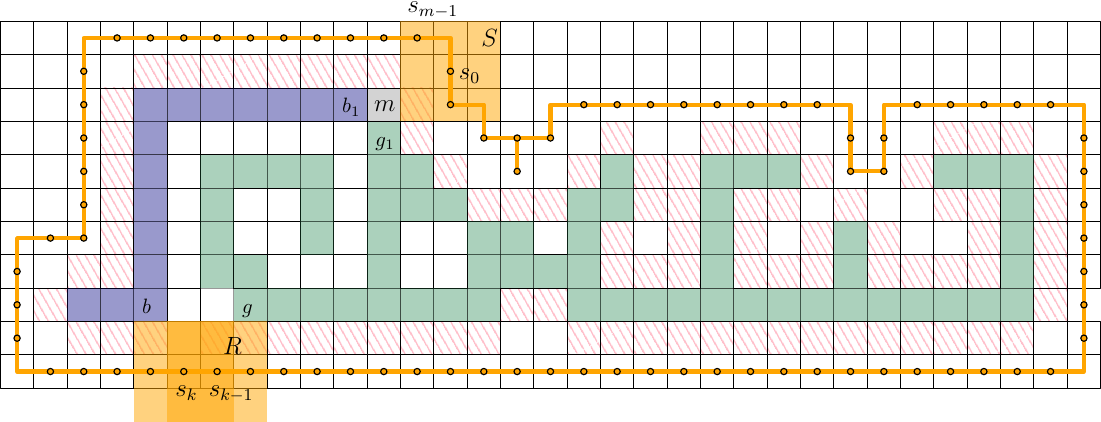}
	\caption{Top: $3\times 3$ square $S$ in its initial position $s_0$. The outer thick line indicates the path traversed by the center of $S$.
		Dots correspond to the center positions where $S$ is adjacent to a boundary edge.
		Bottom: the rectangular union $R$ of $S$ centered at $s_k$ and at $s_{k-1}$.}
	\label{fig:3-square}
\end{figure}

\begin{observation}
	The outer shell has two green-blue changes of color, one happening at $m$. 
\end{observation}

Consider a grid-aligned $3\times 3$ square $S$ centered at the lattice cell of coordinates $s_0=(2,1)$, relative to $m$; see Figure~\ref{fig:3-square} (top).
Translate $S$ orthogonally clockwise one unit at a time along the boundary of the configuration until it reaches $s_0$ again.
Ignoring the positions where $S$ is not adjacent to a boundary edge (i.e., the positions of $S$ where one of its corners coincides
with a convex corner of $C$),
let $s_i$ be the $i$-th position of the center of $S$ along its boundary traversal, and let $s_m=s_0$. 
Refer to Figure~\ref{fig:3-square}.
Since $m$ is the maximum potential module of the configuration, $S$ is empty of modules at position $s_0$ and all subsequent positions, and it does not share edges with the blue component when centered at $s_0$, while at $s_{m-2}$ it is \facet-adjacent to the blue module $b_1$.
Let $s_k$ be the first position of $S$ along its boundary traversal where $S$ becomes \facet-adjacent to a blue module.
Since $S$ travels along the boundary of the configuration, the rectangular union $R$ of $S$ centered at $s_k$ and at $s_{k-1}$ should also be \facet-adjacent to a green module; see Figure~\ref{fig:3-square} (bottom).

The algorithm pivots the \aux\ modules clockwise, following the right-hand rule along the outer shell of the configuration, and brings them to the vicinity of $s_k$ to connect the blue and green components, thus forming a cycle containing $m$.

Let $g$ and $b$ be the closest pair of respectively green and blue modules \facet-adjacent to rectangle $R$, and let $d$ be the $L_1$ distance between them.
The bridging procedure depends on the value of $d$. 
We note that $d>1$ or else both modules would belong to the same component.

\paragraph{Case $d=2$.}

	Up to rotations and reflections, $g$ and $b$ must occur in one of the two configurations shown in Figure~\ref{fig:dist-2} (left).
	\begin{figure}[tbh]
		\centering
		\includegraphics[page=22,width=.81\textwidth]{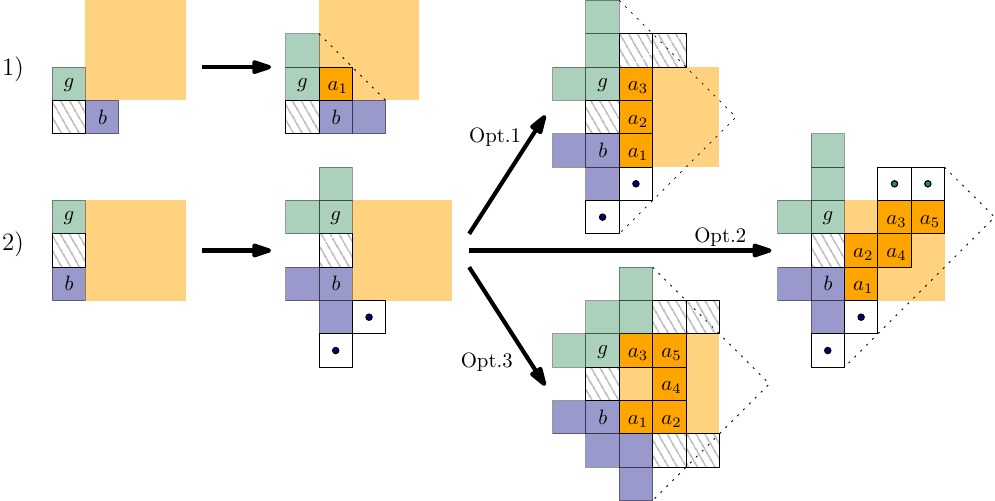}
		\caption{Bridging when $d=2$.
			Striped cells are empty, filled cells are occupied.
			Green and blue indicate different connected components of $C\setminus\{m\}$.
			Dotted pairs of cells of the same color mean that a module must exist in at least one of the two cells.
			The orange cells labelled $a_i$ correspond to the positions of the bridging \aux\ modules. Their subindices indicate the order of their appearance.}
		\label{fig:dist-2}
	\end{figure}
	In the first case, the position North of $g$ (respectively, East of $b$) must be occupied or else $g$ (resp., $b$) would be outer-free, contradicting the precondition of the bridging operation that the configuration contains no outer-free modules.
	Therefore, placing one \aux\ module $a_1$ East of $g$ connects the components without
	changing the maximum and minimum potential of the configuration,
	as shown in the first row of Figure~\ref{fig:dist-2}.
	
	In the second case, the position between $g$ and $b$ must be empty, or else they would belong to the same connected component.
	The positions to the West and North (resp., South) of $g$ (resp., $b$) must contain a module, or else $g$ (resp., $b$) would be outer-free,
	contradicting the precondition of the bridging operation.
	
	The positions shown as blue dots in the figure must contain at least one module or else the module South of $b$ would be outer-free.
	As for the existence of modules on the top and bottom sides of $S$, there are three options.
	\begin{description}
		\item[Option 1:] Both positions $g+(1,1)$ and $g+(2,1)$ are empty. Then we can bridge the two components by sending three \aux\ modules in the order shown in Figure~\ref{fig:dist-2}.
		Notice that if the positions $b+(1,-1)$ and $b+(2,-1)$ are empty, an $x$-symmetric solution applies.
		\item[Option 2:] At least one of positions $g+(2,1)$ or $g+(3,1)$ contains a module. Such module(s) necessarily belong to the green component, by definition of $s_k$.
		If $g+(2,1)$ contains a module, the sequence of four \aux\ modules $a_1$--$a_4$ shown in Figure~\ref{fig:dist-2} can be sent to connect the green and blue components.
		Otherwise, the five \aux\ modules $a_1$--$a_5$ do the bridging.
		\item[Option 3:] Otherwise, the position North-East of $g$ is occupied, while positions $g+(2,1)$ and $g+(3,1)$ are empty. Therefore, position $g+(1,2)$ must be occupied too, or else the module North-East of $g$ would be outer-free. By symmetry, the position South-East of $b$ is occupied, and so is $b+(1,-2)$, while $b+(2,-1)$ and $b+(3,-1)$ are empty. In this case, a sequence of moves can bridge the components using the five \aux\ modules $a_1$--$a_5$ as depicted in Figure~\ref{fig:dist-2}.
	\end{description}
	
	In all cases, the diagonal dotted lines are used to indicate the existence of modules with greater and smaller potential
	than that of the \aux\ modules.
	Note that this bridging operation does not change the maximum and minimum potential of the configuration.
	
\paragraph{Case $d=3$.}

	Up to rotations and reflections, $g$ and $b$ must occur in one of the two configurations shown in Figure~\ref{fig:dist-3} (left).
	\begin{figure}[tbh]
		\centering
		\includegraphics[page=23,width=.6\textwidth]{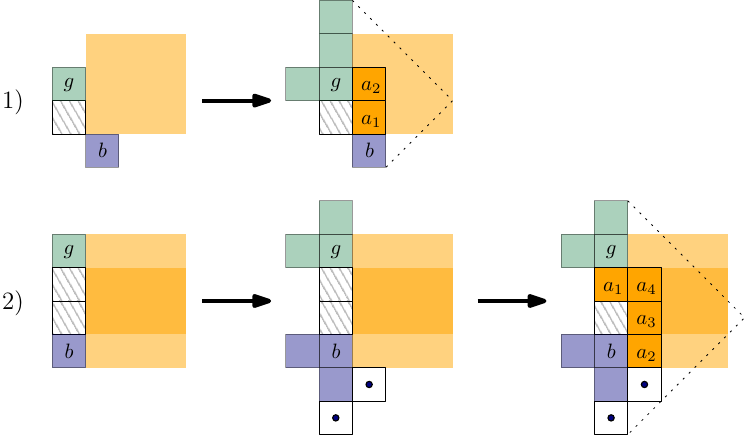}
		\caption{Bridging when $d=3$.
			Striped cells are empty, filled cells are occupied.
			Green and blue indicate different connected components of $C\setminus\{m\}$.
			The dotted pair of cells indicates that a module must exist in at least one of the two cells.
			The orange cells labelled $a_i$ correspond to the positions of the bridging \aux\ modules. Their subindexes indicate the order of their appearance.}
		\label{fig:dist-3}
	\end{figure}
	
	In the first case, the lattice positions West and North of $g$ must be occupied. Therefore, position $g+(0,2)$ must also be occupied. Otherwise the configuration would have a (green) outer-free module. Then a sequence of moves can place the two \aux\ modules $a_1$ and $a_2$ as shown in the first row of Figure~\ref{fig:dist-3}, connecting the green and blue components without
	changing the maximum and minimum potential of the configuration.
	
	In the second case the lattice positions West and North of $g$ must contain a module or $g$ would be outer-free. Symmetrically, the positions West and South of $b$ must also contain a module.
	Furthermore, either $b+(1,-1)$ or $b+(0,-2)$ must also be occupied, or else the module below $b$ would be outer-free.
	Therefore a sequence of moves can place four \aux\ modules $a_1$--$a_4$ as shown in the second row of Figure~\ref{fig:dist-3}. These \aux\ modules connect the green and blue components without changing the
	maximum and minimum potential of the configuration.
	
\paragraph{Case $d=4$.}

	Up to rotations and reflections, $g$ and $b$ must occur in one of the three configurations shown in Figure~\ref{fig:dist-4}.
	\begin{figure}[h]
		\centering
		\includegraphics[page=24,width=.6\textwidth]{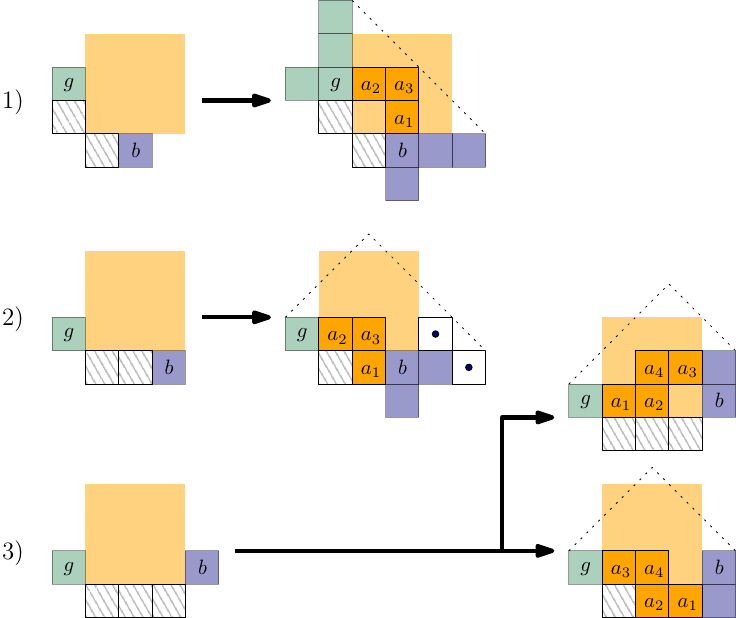}
		\caption{Bridging when $d=4$.
			Striped cells are empty, filled cells are occupied.
			Green and blue indicate different connected components of $C\setminus\{m\}$.
			The dotted pair of cells indicates that a module must exist in at least one of the two cells.
			The orange cells labelled $a_i$ correspond to the positions of the bridging \aux\ modules. Their subindexes indicate the order of their appearance.}
		\label{fig:dist-4}
	\end{figure}
	In the first case, the lattice positions $g+(-1,0)$, $g+(0,1)$ and $g+(0,2)$ must contain a module, or else there would exist an outer-free module. Analogously, $b+(0,-1)$, $b+(1,0)$, and $b+(2,0)$ must be occupied for the same reason.
	Then, a sequence of moves can place three \aux\ modules $a_1$--$a_3$ as shown in the first row of Figure~\ref{fig:dist-4}, connecting the green and blue components without changing the maximum and minimum potential of the configuration.
	In the second case, $b+(0,-1)$, $b+(1,0)$, and  either $b+(2,0)$ or $b+(1,1)$ must be occupied, since the configuration does not have outer-free modules.
	Therefore, a sequence of moves can place three \aux\ modules $a_1$--$a_3$ as shown in the second row of Figure~\ref{fig:dist-4}, connecting the green and blue components without changing the maximum and minimum potential of the configuration.
	Finally, in the third case,
	either $b+(0,-1)$ or $b+(0,1)$ must be occupied. 
	In each case,
	a sequence of moves places four \aux\ modules $a_1$--$a_4$ as shown in the third row of Figure~\ref{fig:dist-4}.
	These \aux\ modules connect the green and blue components without changing the
	maximum and minimum potential of the configuration.
	
\paragraph{Case $d=5$.}

	Up to rotations and reflections, $g$ and $b$ must occur in one of the two configurations shown in Figure~\ref{fig:dist-5}.
	\begin{figure}[h]
		\centering
		\includegraphics[page=25,width=\textwidth]{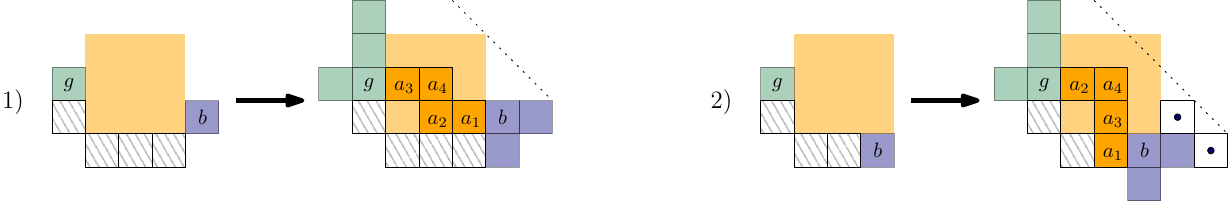}
		\caption{Bridging when $d=5$.
			Striped cells are empty, filled cells are occupied.
			Green and blue indicate different connected components of $C\setminus\{m\}$.
			The dotted pairs of cells indicate that a module must exist in at least one of the two cells.
			The orange cells labelled $a_i$ correspond to the positions of the bridging \aux\ modules. Their subindexes indicate the order of their appearance.}
		\label{fig:dist-5}
	\end{figure}
	In the first case, $g+(-1,0)$, $g+(0,1)$, $g+(0,2)$, $b+(0,-1)$, and $b+(1,0)$ must be occupied, because the configuration has no outer-free modules.
	Then a sequence of moves can place four \aux\ modules $a_1$--$a_4$ as shown in Figure~\ref{fig:dist-5} (left), connecting the green and blue components without changing the maximum and minimum potential of the configuration.
	In the second case, $b+(0,-1)$, $b+(1,0)$, and either $b+(1,1)$ or $b+(2,0)$ must be occupied by modules, since the configuration cannot have outer-free modules.
	Then a sequence of moves can place four \aux\ modules $a_1$--$a_4$ as shown in Figure~\ref{fig:dist-5} (right).
	These \aux\ modules connect the green and blue components without
	changing the maximum and minimum potential of the configuration.
	
\paragraph{Case $d=6$.}

	Up to rotations and reflections, $g$ and $b$ must occur in the configuration shown in Figure~\ref{fig:dist-6}.
	\begin{figure}[h]
		\centering
		\includegraphics[page=26,width=.5\textwidth]{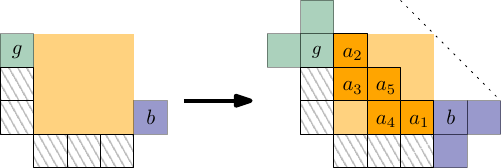}
		\caption{Bridging when $d=6$.
			Striped cells are empty, filled cells are occupied.
			Green and blue indicate different connected components of $C\setminus\{m\}$.
			The orange cells labelled $a_i$ correspond to the positions of the bridging \aux\ modules. Their subindexes indicate the order of their appearance. Note that $a_3$ needs to be moved to the empty space to its left for $a_4$ to be moved to its position.}
		\label{fig:dist-6}
	\end{figure}
	Since the configuration has no outer-free modules, the lattice positions West and North of $g$ must contain modules, and so do the positions East and South of $b$.
	Then a sequence of moves can place five \aux\ modules $a_1$--$a_5$ as shown in Figure~\ref{fig:dist-6}.
	These \aux\ modules connect the green and blue components without changing
	the maximum and minimum potential of the configuration.
	
\paragraph{Bridging result.}
	
	As a result, we obtain the following lemma.
	
\begin{lemma}
	\label{lem:bridge}
	Let $m$ be the North-East module, i.e., the maximum potential module of a given configuration $C$.
	The bridging procedure for $m$ uses $O(n)$ pivoting operations and at most five \aux\ modules,
	and does not change the maximum and minimum potential of the configuration.
	After the procedure ends, $m$ is still the North-East module of the modified configuration, but no longer a cut vertex of its \facet-adjacency graph.
\end{lemma}
	
	\begin{proof}
		It is easy to see that $d$ can only be 2, 3, 4, 5, or 6.
		For each of the possible values of $d$, we have proven that the configuration can be bridged, i.e., that a connection can be made between the two connected components of $C\setminus \{m\}$, using at most five \aux\ modules whose new locations do not increase the potential function of the configuration. Each \aux\ module performs $O(n)$ pivoting operations along the boundary of $C$, for a total of $O(n)$ pivoting operations.
		Since the bridging procedure adds a new connection between the two connected components of $C\setminus \{m\}$, module $m$ is no longer a cut vertex.
		Since the bridging procedure does not increase the potential of the configuration, $m$ is still its maximum potential module.
	\end{proof}

Notice also that the bound of five \aux\ modules for bridging is tight: Figure~\ref{fig:tight-bound} shows an example requiring five \aux\ modules for bridging. 

\begin{figure}[htb]
	\centering
	\includegraphics[page=27,width=.35\textwidth]{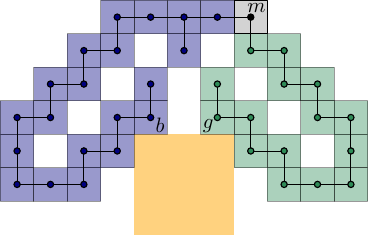}
	\caption{A rigid configuration that requires the addition of five \aux\ modules for bridging.}
	\label{fig:tight-bound}
\end{figure}

\subsubsection{Reconfiguration step}

We now need to guarantee that module $m$ is able to move
and thus, it can pivot along the outer shell of $C$
and join the canonical strip.
This is clear when $m$ is disjoint from the neighborhood of $m$.
We also want to show that we can liberate and send to the canonical strip either the \auxes\ used or at least as many modules as \auxes\ used. 
We will extend the analysis of the neighborhood of $m$,
and for each of the possible cases we will show that
either invoking the bridging procedure or
explicitly placing \aux\ modules
 can guarantee that.

Progress is measured in terms of the potential gap $\Delta\Phi=\Phi_{max}-\Phi_{min}$ of the configuration and the size of $C$ (recall that $C$ includes all modules that are not part of the canonical strip).
We will show that each reconfiguration step decreases the potential gap and/or the size of $C$.

Recall the setting of the configuration before performing a reconfiguration step.
Module $m$ is the North-East module (i.e., the highest module of maximum potential in a given configuration $C$),
it has degree $2$ and connects two connected components of $C$, one blue (extending counterclockwise) and one green (extending clockwise).
Let $b_1$ and $g_1$ be the blue and green modules adjacent to $m$, respectively.
Position $m +(-1,-1)$ must be empty, since otherwise the green and blue modules would belong to the same component. 

\begin{lemma}\label{lem:b1g1-degree}
	The neighborhood around $m$ is in one of the three configurations shown in Figure~\ref{fig:b1g1-degree} (b1--b3): its South neighbor $g_1$ has degree $2$ and a South neighbor $g_2$, and its West neighbor $b_1$ has maximum degree $2$.
\end{lemma}
\begin{proof}
	By the definition of $m$, the positions $m +(1,0)$, $m +(0,1)$, and $m+(1,-1)$ must be empty, otherwise $C$ would contain a module with potential higher than that of $m$, a contradiction.
	
	\begin{figure}[tbh]
		\centering
		\includegraphics[page=28,width=.8\textwidth]{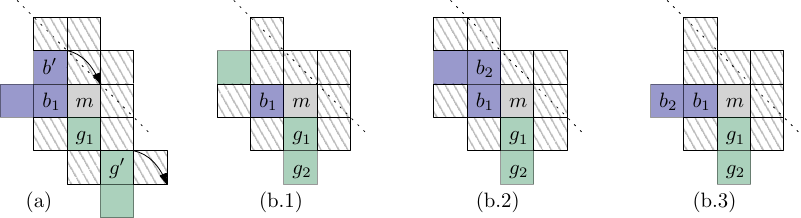}
		\caption{(a) If $g'$ exists, it is outer-free; if $b_1$ is of degree $3$, $b'$ is outer-free (b.1) $b_1$ has degree $1$
			(b.2) $b_1$ has a North neighbor (b.3) $b_1$ has a West neighbor.}
		\label{fig:b1g1-degree}
	\end{figure}
	
	Since the positions West and East of $g_1$ are empty, $g_1$ has maximum degree $2$. Assume first that $g_1$ has degree $1$. In this case the position $g_1 + (1,-1)$ must be occupied by a module $g'$, otherwise $g_1$ would be an outer-free module.
	Refer to Figure~\ref{fig:b1g1-degree} (a). Note that the positions West, North, and East of $g'$ are empty (the last two due to the definition of $m$), which means that $g'$ has degree $1$ and can pivot clockwise without disconnecting $C$, contradicting the fact that $C$ has no outer-free modules. It follows that $g_1$ is of degree $2$ and has a South neighbor.
	
	Since the position below $b_1$ is empty, $b_1$ has maximum degree $3$. Assume first that $b_1$ has degree $3$, as shown in Figure~\ref{fig:b1g1-degree} (a). Note, however, that the positions above and to the right of its North neighbor $b'$ are all empty (because these positions have potential higher than that of $m$, which is maximum) and therefore $b'$ can pivot clockwise without disconnecting $C$, contradicting the fact that $C$ has no outer-free modules. Thus $b_1$ has degree at most $2$.
	
	We now argue that the neighborhood around $m$ is in one of the three configurations shown in Figure~\ref{fig:b1g1-degree} (b.1--b.3). If $b_1$ has degree $1$, then the position $b_1 + (-1,1)$ must be occupied, otherwise $b_1$ would be outer-free; see Figure~\ref{fig:b1g1-degree} (b.1). 
	Note that this position must be occupied by a green module, otherwise the blue component would be disconnected.
	If $b_1$ has a North neighbor, then the position
	$b_1 + (-1,1)$ must also be occupied, otherwise the North neighbor would be outer-free; see Figure~\ref{fig:b1g1-degree} (b.2). 
	Note that this position must be occupied by a blue module, otherwise the green and blue components would be connected.
	The only case left is that $b_1$ has degree 2 and its neighbor is to its West; see Figure~\ref{fig:b1g1-degree} (b.3).
\end{proof}

Let $g_2$ be the South neighbor of $g_1$ guaranteed by Lemma~\ref{lem:b1g1-degree}, and let $b_2$ be the North or West neighbor of $b_1$ (if one exists). Our reconfiguration procedure depends on the degrees of $g_2$, $b_1$ and $b_2$ (if it exists).

\paragraph{Reconfiguring when $b_1$ has degree 1.}

\begin{figure}[tbh]
	\centering
	\includegraphics[page=29,width=.9\textwidth]{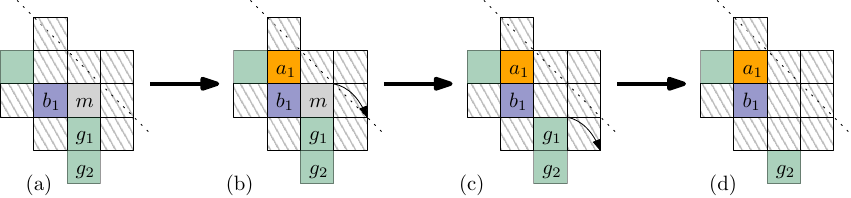}
	\caption{Reconfiguring when $b_1$ has degree $1$.}
	\label{fig:reconf-b1degree1}
\end{figure}
This case is depicted in Figures~\ref{fig:b1g1-degree} (b.1) and~\ref{fig:reconf-b1degree1}.
In this case we simply place a \aux\ module above $b_1$; it is labelled $a_1$ in Figure~\ref{fig:reconf-b1degree1}. 
This \aux\ module connects the blue and green components, leaving $m$ free to pivot along the outer shell of $C$ and join the canonical strip; see Figure~\ref{fig:reconf-b1degree1} (b). 
Module $g_1$ then becomes outer-free and can pivot clockwise to join the canonical strip; see Figure~\ref{fig:reconf-b1degree1} (c).

\begin{lemma}\label{lem:r1}
	This reconfiguration step uses $O(n)$ pivoting operations to transform $C$ into \afacet-connected configuration of smaller size and smaller potential gap $\Delta\Phi$.
\end{lemma}
\begin{proof}
	In addition to the $O(n)$ pivoting operations used by $m$ to join the canonical strip, this reconfiguration step uses only a constant number of pivoting operations, as indicated by Figure~\ref{fig:reconf-b1degree1}. The output of this reconfiguration step, shown in Figure~\ref{fig:reconf-b1degree1} (d), is \facet-connected, because each of the blue and green components are \facet-connected, and $a_1$ \facet-connects to each of them. Because the \aux\ module $a_1$ joins $C$ and the configuration modules $m$ and $g_1$ leave $C$, both the size and the maximum potential of $C$ decrease as a result. Note that the minimum potential of $C$ stays the same, therefore the potential gap $\Delta\Phi$ decreases.
\end{proof}

\noindent
From this point on we assume that $b_1$ has degree $2$ and therefore $b_2$ exists.

\paragraph{Reconfiguring when $g_2$ or $b_2$ has degree 1.}

First observe that, if $g_2$ has degree one, then the green component consists of $g_1$ and $g_2$ only (otherwise the green component would be disconnected). Also note that position $g_2 + (1, -1)$ must be occupied by a blue module, otherwise $g_2$ would be outer-free. We handle this situation similarly to the one above: place a \aux\ module (labelled $a_1$ in Figure~\ref{fig:reconf-g2degree1}) to the right of $g_2$, to connect the blue and green components. This leaves $m$ free to pivot along the outer shell of $C$  and join the canonical strip; see Figure~\ref{fig:reconf-g2degree1} (b). It is then followed by $g_1$; see Figure~\ref{fig:reconf-g2degree1} (c). The argument is similar for the case when $b_2$ exists and has degree~$1$.

\begin{figure}[tbh]
	\centering
	\includegraphics[page=30,width=.85\textwidth]{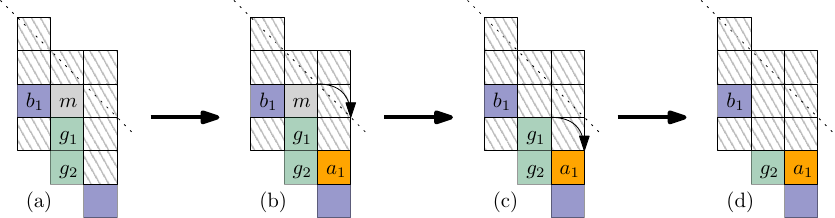}
	\caption{Reconfiguring when $g_2$ and $b_2$ have degree $2$.}
	\label{fig:reconf-g2degree1}
\end{figure}

\begin{lemma}\label{lem:r2}
	This reconfiguration step uses $O(n)$ pivoting operations to transform $C$ into \afacet-connected configuration of smaller size
	and smaller potential gap $\Delta\Phi$.
\end{lemma}
\begin{proof}
	In addition to the $O(n)$ pivoting operations used by $m$ to join the canonical strip, this reconfiguration step uses only a constant number of pivoting operations, as indicated by Figure~\ref{fig:reconf-g2degree1}. The output of this reconfiguration step is \facet-connected, because each of the blue and green components are \facet-connected, and $a_1$ \facet-connects to each of them; see  Figure~\ref{fig:reconf-g2degree1} (d). Because the \aux\ module $a_1$ joins $C$ and the configuration modules $m$ and $g_1$ leave $C$, both the size and the maximum potential of $C$ decrease as a result. Note that the minimum potential of $C$ remains unchanged, therefore the potential gap $\Delta\Phi$ decreases.
\end{proof}

\paragraph{Reconfiguring when $g_2$ and $b_2$ have degree 2.}

We start with the configurations guaranteed by Lemma~\ref{lem:b1g1-degree} and depicted in Figures~\ref{fig:b1g1-degree} (b.2--b.3), and place the second neighbor of $b_2$ and $g_2$ in all possible positions. The result for $b_2$ ($g_2$, resp.) is depicted on the top (bottom, resp.) of  Figure~\ref{fig:reconf-b2g2degree2}.
\begin{figure}[tbh]
	\centering
	\includegraphics[page=31,width=.85\textwidth]{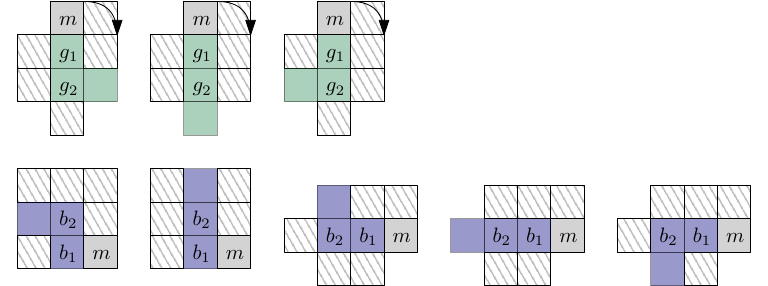}
	\caption{Reconfiguring when $g_2$ and $b_2$ have degree $2$.}
	\label{fig:reconf-b2g2degree2}
\end{figure}
Notice that $b_2$ cannot have a North neighbor since it would have a higher potential than $m$.
In this case we invoke the bridging procedure to connect the green and blue components with a set of \aux\ modules, so that $m$ becomes outer-free and can pivot clockwise to join the canonical strip. For each combination of (top, bottom) configurations depicted in Figure~\ref{fig:reconf-b2g2degree2}, $g_1$, $g_2$, $b_1$ and $b_2$ become outer-free and pivot clockwise to  join the canonical strip in this order.

\begin{lemma}\label{lem:r3}
	This reconfiguration step uses $O(n)$ pivoting operations to transform $C$ into \afacet-connected configuration of smaller potential gap $\Delta\Phi$.
\end{lemma}
\begin{proof}
	By Lemma~\ref{lem:bridge}, the bridging procedure uses at most five \aux\ modules and $O(n)$ pivoting operations. The reconfiguration itself uses $O(n)$ pivoting operations to send five modules ($m$, $b_1$, $b_2$, $g_1$ and $g_2$) to the canonical strip. The resulting configuration is \facet-connected because each of the blue and green components are \facet-connected, and the bridge (which remains in place) \facet-connects to each of them. Since at most five \aux\ modules join $C$ and exactly five others leave $C$, the size of $C$ stays the same. However, since $m$ leaves $C$, the maximum potential of $C$ decreases. Note that the minimum potential of $C$ stays the same, therefore the potential gap $\Delta\Phi$ decreases.
\end{proof}

\noindent
It remains to discuss the situations when at least one of $b_2$ and $c_2$ has degree strictly greater than $2$.

\paragraph{Reconfiguring when $g_2$ has degree greater than 2.}

If $g_2$ does not have a West neighbor, then $g_2$ has East and South neighbors and is of degree $3$; see Figure~\ref{fig:reconf-g2degree3} (a). Note however that in this case the East neighbor of $g_2$ is outer-free (because the positions above and to its right have potential higher than that of $m$ and are therefore empty), contradicting the fact that $C$ has no outer-free modules. This implies that $g_2$ has a West neighbor. In this case we invoke the bridging procedure to connect the green and blue components joined by the cut vertex $m$ with a set of \aux\ modules, so that $m$ becomes outer-free and can pivot clockwise to join the canonical strip.
Recall that we are in the context where $b_1$ has degree $2$ before reconfiguration (Lemma~\ref{lem:b1g1-degree} guarantees that the degree of $b_2$ is at most $2$, and the case with $b_1$ of degree $1$ has been handled above), so $b_1$ has degree $1$ after $m$ rolls away. By Lemma~\ref{lem:b1g1-degree}, the neighbor $b_2$ of $b_1$ lies to the North or West of $b_1$, as indicated by the two configurations from Figures~\ref{fig:b1g1-degree} (b.2--b.3).

Assume first that $b_2$ lies North of $b_1$. Refer to Figure~\ref{fig:reconf-g2degree3} (b). We discuss two situations, depending on whether the position $b' = b_1 +(-1,-1)$ is empty or not. If position $b'$ is empty, then a sequence of pivoting operations reconnects the green and blue components in the vicinity of $b_1$ as follows. 
First, $b_1$ pivots counterclockwise to attach to $g_1$; see Figure~\ref{fig:reconf-g2degree3} (c.1). 
Second, $b_2$ pivots clockwise to attach on top of $b_1$; see  Figure~\ref{fig:reconf-g2degree3} (c.2). 
Finally, $g_1$ pivots counterclockwise twice to attach North of $b_2$; see  Figure~\ref{fig:reconf-g2degree3} (c.3). The result is shown in Figure~\ref{fig:reconf-g2degree3} (c.4).

\begin{figure}[tbh]
	\centering
	\includegraphics[page=32,width=\textwidth]{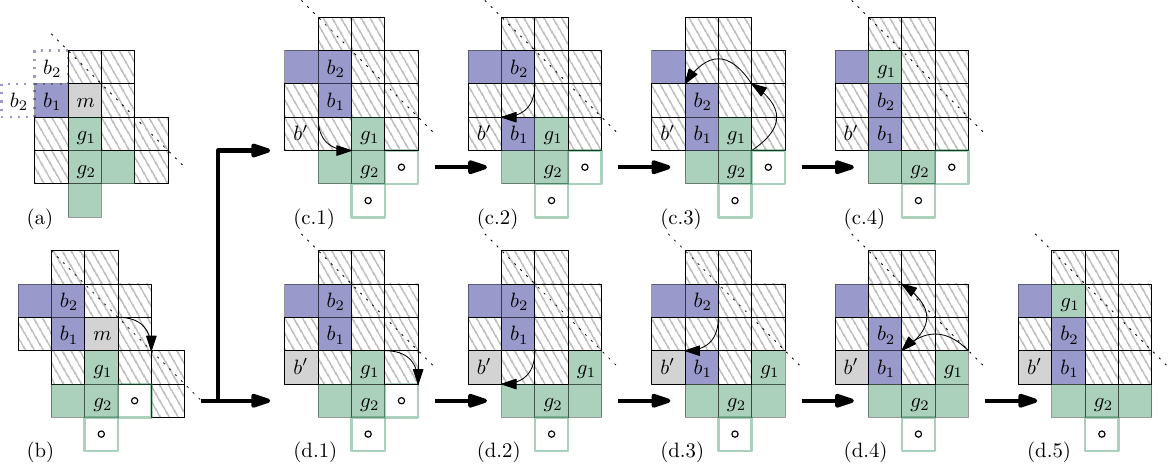}
	\caption{Reconfiguring when $g_2$ has degree greater than $2$ (a) $g_2$ does not have a West neighbor (b)  $g_2$ has a West neighbor and $b_2$ lies North of $b_1$.}
	\label{fig:reconf-g2degree3}
\end{figure}

If the position $b'$ is occupied, then $b_1$ is blocked; see Figure~\ref{fig:reconf-g2degree3} (d.1). In this case $g_1$ pivots clockwise to free the space South-East of $b_1$. Figure~\ref{fig:reconf-g2degree3} (d.2) shows the case when $g_1$ attaches to the East neighbor of $g_2$, but the argument holds for the case when $g_2$ does not have an East neighbor (in this case $g_1$ would attach East of $g_2$). With $g_1$ out of the way, $b_1$ can now pivot clockwise and attach East of $b'$, followed by $b_2$ which pivots clockwise to attach North of $b_1$; see Figure~\ref{fig:reconf-g2degree3} (d.3). Finally, $g_1$ reverses its pivoting step back to its original position, then pivots counterclockwise once more to attach North of $b_2$; see Figure~\ref{fig:reconf-g2degree3} (d.4). The result is shown in Figure~\ref{fig:reconf-g2degree3} (d.5).
\begin{figure}[tbh]
	\centering
	\includegraphics[page=33,width=\textwidth]{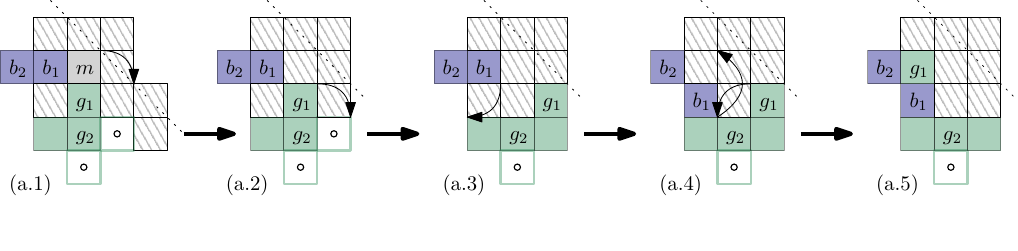}
	\caption{Reconfiguring when $g_2$ has degree greater than $2$, $g_2$ has a West neighbor and $b_2$ lies West of $b_1$.}
	\label{fig:reconf-g2degree3-2}
\end{figure}
Assume now that $b_2$ lies West of $b_1$; see Figure~\ref{fig:reconf-g2degree3-2} (a.1). In this case, after $m$ rolls away to join the canonical strip, a sequence of pivoting operations reconnects the green and blue components in the vicinity of $b_1$ as follows. 
First, $g_1$ pivots clockwise; see Figure~\ref{fig:reconf-g2degree3-2} (a.2). 
Second, $b_1$ pivots clockwise; see Figure~\ref{fig:reconf-g2degree3-2} (a.3). 
Finally, $g_1$ reverses its pivoting step back to its original position, then pivots counterclockwise once more to attach North of $b_1$; see Figure~\ref{fig:reconf-g2degree3-2} (a.4). The result is shown in Figure~\ref{fig:reconf-g2degree3-2} (a.5).

In all these cases, the green and blue components remain connected after the \aux\ modules retrace their steps to join the canonical strip.

\begin{lemma}\label{lem:r4}
	This reconfiguration step uses $O(n)$ pivoting operations to transform $C$ into \afacet-connected configuration of smaller size and smaller potential gap $\Delta\Phi$.
\end{lemma}
\begin{proof}
	By Lemma~\ref{lem:bridge}, the bridging procedure (and its reverse) takes $O(n)$ pivoting operations. In addition to the $O(n)$ pivoting operations used by $m$ to join the canonical strip, this reconfiguration step uses only a constant number of pivoting operations. The resulting configuration is \facet-connected, because each of the blue and green components are connected, and this reconfiguration step connects the blue and green components together, as shown in the right columns of
	Figures~\ref{fig:reconf-g2degree3} and~\ref{fig:reconf-g2degree3-2}.
	Because the \aux\ modules rejoin the canonical strip, the size of $C$ decreases. The maximum potential of $C$ also decreases (because $m$ leaves $C$) and the minimum potential of $C$ stays the same, therefore the potential gap $\Delta\Phi$ decreases.
\end{proof}

\paragraph{Reconfiguring when $b_2$ has degree greater than 2.}

By Lemma~\ref{lem:b1g1-degree}, the neighbor $b_2$ of $b_1$ lies North or West of $b_1$, as indicated by the two configurations from Figures~\ref{fig:b1g1-degree} (b.2--b.3). Note however that, if $b_2$ lies North of $b_1$, then the positions North and East of $b_2$ are empty (since their potential is higher than that of $m$), which implies that $b_2$ has degree $2$. So the only situation left to discuss here is when $b_2$ lies West of $b_1$.

\begin{figure}[tbh]
	\centering
	\includegraphics[page=34,width=0.95\textwidth]{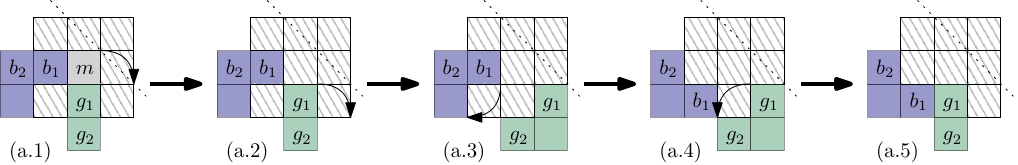}
	\caption{Reconfiguring when $b_2$ has degree greater than $2$ and a South neighbor.}
	\label{fig:reconf-b2degree3}
\end{figure}

Assume first that $b_2$ has a South neighbor; see Figure~\ref{fig:reconf-b2degree3}. 
This case is very similar to the one where $g_2$ has a West neighbor (discussed in the previous section and depicted in Figure~\ref{fig:reconf-g2degree3-2}), and the same sequence of operations applies here as well: after bridging, $m$ pivots clockwise along the outer shell to join the canonical line; $g_1$ and $b_1$ pivot clockwise, in this order; then $g_1$ pivots counterclockwise (the only difference is that $g_1$ stops after the first pivoting step). This sequence of pivoting steps is depicted in Figure~\ref{fig:reconf-b2degree3}.

\begin{figure}[h]
	\centering
	\includegraphics[page=35,width=0.94\textwidth]{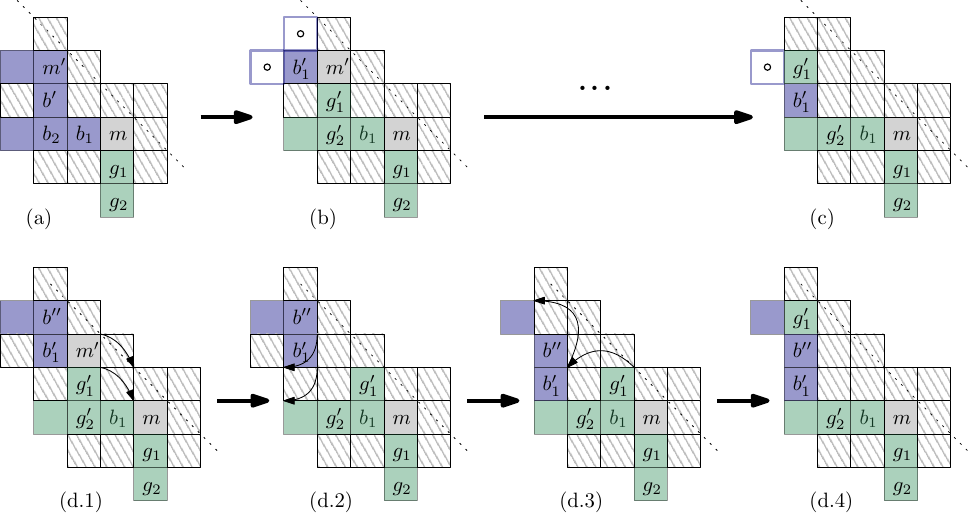}
	\caption{Reconfiguring when $b_2$ has degree greater than $2$ but no South neighbor.}
	\label{fig:reconf-b2degree3-2}
\end{figure}

Assume now that $b_2$ does not have a South neighbor. Since the degree of $b_2$ is at least $3$, $b_2$ must have West and North neighbors.
Refer to Figure~\ref{fig:reconf-b2degree3-2} (a).
Let $b'$ be the North neighbor of $b_2$ and note that the position $m'$ North of $b'$ must be occupied, otherwise $b'$ would be outer-free. Also note that the position West of $m'$ must be occupied and the position West of $b'$ must be empty, otherwise $m'$ would be outer-free. We reassign the role of $m$ to $m'$ and recolor the graph so that its blue and green components are joined by $m'$; see Figure~\ref{fig:reconf-b2degree3-2} (b). Let $b'_1$ and $g'_1$ be the blue and green neighbors of $m'$.
If $b'_1$ does not have a North neighbor, we are in a situation similar to the one depicted in Figure~\ref{fig:reconf-g2degree3-2} and handle it the same way: bridge the (new) blue and green components;
pivot $m'$ clockwise along the outer shell to join the canonical line;
pivot $g'_1$ and $b'_1$ clockwise, in this order;
then pivot $g'_1$ counterclockwise twice.
The result is shown in Figure~\ref{fig:reconf-b2degree3-2} (c).

We use a similar sequence of pivoting steps for the case when $b'_1$ has a North neighbor $b''$; see Figure~\ref{fig:reconf-b2degree3-2} (d.1).
First note that the position North of $b''$ is empty, since it has potential higher than the one of $m$. Also note that the position West of $b''$ is occupied and the position West of $b'_1$ is empty, otherwise $b''$ would be outer-free. This enables the following sequence of operations. 
First, we bridge the (new) blue and green components. 
Second, we pivot $m'$ clockwise along the outer shell to join the canonical line. 
Then, we pivot $g'_1$, $b'_1$ and $b''$ clockwise, in this order; see Figure~\ref{fig:reconf-b2degree3-2} (d.2). 
Finally, we pivot $g'_1$ counterclockwise thrice; see Figure~\ref{fig:reconf-b2degree3-2} (d.3).
The result is shown in Figure~\ref{fig:reconf-b2degree3-2} (d.4).

In all these cases, the green and blue components remain connected after the \aux\ modules retrace their steps to join the canonical line.

\begin{lemma}\label{lem:r5}
	This reconfiguration step uses $O(n)$ pivoting operations to transform $C$ into \afacet-connected configuration of smaller size
	and smaller or equal potential  gap $\Delta\Phi$.
\end{lemma}
\begin{proof}
	By Lemma~\ref{lem:bridge}, the bridging procedure (and its reverse) takes $O(n)$ pivoting operations. In addition to the $O(n)$ pivoting operations used by $m$ or $m'$ to join the canonical strip, this reconfiguration step uses only a constant number of pivoting operations.
	The resulting configuration is \facet-connected, because each of the blue and green components are connected, and this reconfiguration step connects the blue and green components together, as shown in the right columns of Figures~\ref{fig:reconf-b2degree3} and~\ref{fig:reconf-b2degree3-2}.
	Note that the size of $C$ decreases, because the \aux\ modules rejoin the canonical strip.
	In both cases the potential gap does not increase.
\end{proof}

\subsection{Algorithm pseudocode}
Algorithm~\ref{alg} solves the reconfiguration problem by combining the operations described in the previous sections:

\begin{algorithm}[h]
	\KwData{An arbitrary \facet-connected configuration $C$ with $n$ modules}
	\KwResult{A canonical strip of modules of length $n$}
	\While{there are still modules in $C$}{
		\While{there exist outer-free modules}{
			pick one outer-free module and pivot it all the way to the tip of the strip\;
		}
		\If{the strip has fewer than five modules}{
			make the strip five modules long by adding \aux\ modules\;
		}
		invoke the reconfiguration step\;
	}
	\caption{Reconfiguring an arbitrary \facet-connected configuration into a canonical strip.}
	\label{alg}
\end{algorithm}

\begin{theorem}
	The reconfiguration algorithm (Algorithm~\ref{alg})
	transforms \afacet-connected configuration $C$ with $n$ modules
	into a canonical strip of the same size,
	using $O(n^2)$ monkey-move pivoting steps, which is worst-case optimal, and
	adding at most five extra modules.
\end{theorem}
\begin{proof}
	The input to the algorithm is a configuration $C$ of size $n$ and potential gap $\Delta\Phi =O(n)$. Each step of the innermost loop uses $O(n)$ pivoting
	operations to take an outer-free module to the end of the strip, thus decreasing the size of $C$ by one.
	Each reconfiguration step uses $O(n)$ pivoting steps to
	decrease either the potential gap
	or the size of $C$, leaving it \facet-connected (and never increasing the potential gap).
	Because the size of $C$ never increases, the length of the canonical strip never decreases. This means that
	the strip can have fewer than five modules only once and the conditional does not affect the complexity of the algorithm. 
	We conclude that the algorithm terminates after $O(n)$ iterations in total.
	Because each iteration takes $O(n)$ pivoting steps, the total number of pivoting steps is $O(n^2)$.
	Optimality comes from the $\Omega(n^2)$ pivoting steps required to reconfigure a vertical strip into a horizontal one.
\end{proof}

\section{Conclusion and open problems}\label{sec:conclusions}
This paper addresses the problem of reconfiguring \afacet-connected grid configuration of $n$ modules into any other configuration of $n$ modules under three increasingly more flexible sets of pivoting moves, namely restrictive, leapfrog and monkey.
Previous results solve this problem under the leapfrog set of moves, as long as the initial and final configurations satisfy a strong local separating condition imposed by three forbidden patterns. We
show that there exist robot configurations with many instances of the three forbidden patterns that are still reconfigurable, so the local separation condition is not necessary. On the other hand, we show that as soon as the local separation condition is relaxed, the reconfiguration graph breaks into an exponential number of connected components of exponential size. To overcome this obstacle we introduce a new pivoting move, called monkey, and a natural reconfiguration approach that does not depend on local features, but uses up to five extra modules that can freely move around the boundary of the robot configuration. These extra modules are used to unlock intermediate locked configurations so that progress can be made towards the target configuration. We show that our approach uses $O(n^2)$ monkey-pivoting moves to reconfigure any source configuration with $n$ pivoting modules into any given target configuration.

We leave open the question of whether universal reconfiguration can be accomplished under the more restrictive set of leapfrog pivoting moves using a constant number of extra modules.

Another question is whether our approach generalizes to three or higher dimensions.
For example, when the slice graphs (where the vertices are the \emph{slices} of the configuration cut along an axis and the edges connect slices with \facet-adjacent modules) of the source and target configurations are both paths, we should be able to reconfigure each to a strip of modules, one slice at a time, similar to our 2-dimensional approach does. We conjecture that a similar approach will also work for general 3-dimensional configurations, potentially after increasing the number of \aux\ modules to bridge larger gaps introduced by the higher dimensionality.

\bibliography{Pivoting-facet-adjacent}

\end{document}